\documentclass[aps,pra,groupedaddress,superscriptaddress,amsmath,amssymbol,notitlepage,stfloats]{revtex4-1}

\usepackage[latin1]{inputenc}
\usepackage{amsthm}
\usepackage{amssymb}
\usepackage{graphicx}
\usepackage{amsmath}
\usepackage{bbold}
\usepackage{bbm}
\usepackage{nonfloat}
\usepackage{color}
\usepackage[pdftex]{hyperref}
\usepackage{braket}
\usepackage{graphicx}
\usepackage{dsfont}
\usepackage{mathdots}
\usepackage{enumitem}
\usepackage{tcolorbox}

\DeclareGraphicsExtensions{.jpeg, .png, .pdf}

\newtheorem{thm}{Theorem}
\newtheorem*{thm*}{Theorem}
\newtheorem{prop}[thm]{Proposition}
\newtheorem*{prop*}{Proposition}

\newtheorem*{lemma*}{Lemma}
\newtheorem{cor}[thm]{Corollary}
\newtheorem*{cor*}{Corollary}

\newtheorem*{ex*}{Example}

\newtheorem*{cj*}{Conjecture}
\newtheorem{Def}[thm]{Definition}
\newtheorem*{Def*}{Definition}

\theoremstyle{definition}
\newtheorem*{rem}{Remark}

\newcommand{\bq}{\begin{equation*}}
\newcommand{\be}{\begin{equation}}
\newcommand{\eq}{\end{equation*}}
\newcommand{\ee}{\end{equation}}
\newcommand{\bmu}{\begin{multline*}}
\newcommand{\emu}{\end{multline*}}
\newcommand{\ban}{\begin{align*}}
\newcommand{\bal}{\begin{align}}
\newcommand{\ean}{\end{align*}}
\newcommand{\eal}{\end{align}}
\newcommand{\Tr}{\text{Tr}\,}
\newcommand{\rk}{\text{rk}\,}

\tcbset{colback=green!5,colframe=orange!50!red,
fonttitle=\bfseries, float=htb}

\newtcolorbox[auto counter]{example}[3][]
{float*=ht,title=Observation ~\thetcbcounter: #2,label= ex:#3 ,#1}

\begin{document}

\title{Bipartite depolarizing maps}
\author{Ludovico Lami}
\affiliation{Universitat Aut\`onoma de Barcelona, 08193 Bellaterra, Barcelona, Spain}
\author{Marcus Huber}
\affiliation{Group of Applied Physics, University of Geneva, 1211 Geneva 4, Switzerland}
\affiliation{Universitat Aut\`onoma de Barcelona, 08193 Bellaterra, Barcelona, Spain}

\begin{abstract}
We introduce a $3$--parameter class of maps \eqref{Phi} acting on a bipartite system that are a natural generalisation of the depolarizing channel (and include it as a special case). Then, we find the exact regions of the parameter space that alternatively determine a positive, completely positive, entanglement--breaking or entanglement--annihilating map. This model displays a much richer behaviour than the one shown by a simple depolarizing channel, yet it stays exactly solvable. As an example of this richness, PPT but not entanglement--breaking maps are found in Theorem \ref{EB region}. A simple example of a positive yet indecomposable map is provided (see the Remark at the end of Section \ref{sec EB}). The study of the entanglement--annihilating property is fully addressed by Theorem \ref{EA}. Finally, we apply our results to solve the problem of the entanglement annihilation caused in a bipartite system by a tensor product of local depolarizing channels. In this context, a conjecture posed in \cite{EA4} is affirmatively answered, and the gaps that the imperfect bounds of \cite{EA3} left open are closed. To arrive at this result we furthermore show how the Hadamard product between quantum states can be implemented via local operations.
\end{abstract}

\maketitle

\section{Introduction} \label{sec intro}
Entanglement appears as the central pillar of almost every quantum communication protocol \cite{nielsen}. It is however a very fragile resource and once distributed, the local participants in any kind of communication protocol can never increase it via local operations, even if aided by unbounded classical communication (these operations are abbreviated under the acronym LOCC). Much effort has thus been invested into characterizing the resource of entanglement \cite{horodeckiqe,siewert}. It is in general an NP-hard problem to decide whether a given quantum state exhibits entanglement at all \cite{gurvits}, even if the full state is known with perfect precision. This precludes a complete characterization of all aspects in full generality, but nonetheless some paradigmatic cases can be resolved in a practically useful way. Complete solutions for the separability problem exist for various special classes of states \cite{WernerOrig,Iso-dep,VollbrechtWerner,SPA,MaxAbelian,Marcus1,axi,Buchy,SepBos} and through mapping arbitrary states into these classes one can get useful bounds for general states. A central tool for addressing these problems are positive maps, the canonical example being transposition, from which the probably most widely used entanglement test derives \cite{Peres}. Apart from quantifying \cite{siewert} and detecting entanglement \cite{guhnetoth}, another important problem concerns the transformation of entanglement.

In general, any physical transformation of a quantum system can be represented by a completely positive (CP) map. Studying which CP maps preserve and more importantly which maps adversely affect entanglement is of great importance to understanding the limitations of sending quantum states through different channels. Two important examples are maps that break the entanglement between two quantum systems (entanglement--annihilating) and between a single quantum system and the rest of the world (entanglement--breaking). 

Among the most significant quantum channel we can undoubtedly include the depolarizing channel, defined on a $d$--dimensional system by the formula $\Delta_\lambda\equiv \lambda I +(1-\lambda)\frac{\mathds{1}}{d}\Tr$. Its usefulness is twofold: on the one hand, it represents a physically meaningful model of the white noise acting on an isolated system; on the other hand, it constitutes a mathematically treatable yet nontrivial example of quantum channel. We remind the reader \cite{Iso-dep} that the map $\Delta_\lambda$ is:
\begin{itemize}
\item positive iff $-\frac{1}{d-1}\leq \lambda\leq 1$;
\item completely positive (CP) iff $-\frac{1}{d^2-1}\leq \lambda\leq 1$;
\item completely copositive (coCP) iff $-\frac{1}{d-1}\leq \lambda\leq \frac{1}{d+1}$;
\item entanglement--breaking iff it is at the same time CP and coCP, iff $-\frac{1}{d^2-1}\leq\lambda\leq \frac{1}{d+1}$. 
\end{itemize}
In particular, it is useful to remember that whenever $\Delta_\lambda$ is positive it is also either completely positive or completely copositive (or both). These facts are used several times in the rest of the paper.

In this manuscript we study a three--parameter class of maps that generalises the depolarizing channel, one of the most studied transformations in the context of noisy communication. Despite the richer structure that class displays, we present an exact solution in terms of all parameters $\alpha, \beta, \gamma$ for which this map is \emph{positive}; \emph{completely positive}; \emph{entanglement--breaking} and \emph{entanglement--annihilating}. This characterization includes a new example of a positive map, which is not decomposable, i.e. capable of detecting entangled states which remain positive under partial transposition (PPT). To arrive at this map we point out a physically interesting Lemma: The Hadamard product between two arbitrary quantum states can be stochastically implemented via LOCC.

As an application of the entanglement--annihilating characterization, we answer some open questions recently raised in \cite{EA3,EA4} concerning the minimum amount of local depolarizing noise leading to a global entanglement annihilation when applied on both sides of a bipartite system. \\

To start, let us consider the following maps acting on a bipartite system $AB$:
\be \Phi[\alpha,\beta,\gamma]\ \equiv\ \mathds{1}_{AB} \Tr +\, \alpha\, \mathds{1}_{A} \Tr \otimes I\, +\, \beta\, I\otimes \mathds{1}_{B} \Tr +\, \gamma\, I\ . \label{Phi} \ee

As a first remark, one could wish to consider a more general linear combination than \eqref{Phi}, with another parameter for the coefficient of $\mathds{1}_{AB} \Tr$. However, already the positivity condition easily implies that such a parameter must be non--negative. Up to a normalization constant, one can take it to be either $1$ or $0$. The latter case can be deduced from the former, because $\alpha\, \mathds{1}_{A} \Tr \otimes I\, +\, \beta\, I\otimes \mathds{1}_{B} \Tr +\, \gamma\, I$ is positive iff $\frac{1}{M}\, \mathds{1}_{AB}\Tr + \alpha\, \mathds{1}_{A} \Tr \otimes I\, +\, \beta\, I\otimes \mathds{1}_{B} \Tr +\, \gamma\, I$ is positive for all $M>0$, iff $\Phi[M\alpha,M\beta,M\gamma]$ is positive for all $M>0$.

Observe that the maps $\Phi$ we defined commute with all the local unitary conjugations. But there is more: up to including a coefficient for the first addend (which is an irrelevant modifications in the sense clarified above), these maps are \emph{all} the maps on $AB$ displaying this feature, as the following reasoning shows. Before going into the details, we remind the reader that the states of a bipartite system $AA'$ commuting with all the conjugations by local unitaries of the form $U_A\otimes U^*_{A'}$ are exactly the linear combinations of the identity $\mathds{1}_{AA'}$ and the maximally entangled state $\ket{\varepsilon}\!\!\bra{\varepsilon}_{AA'}$. Now, up the application of the Choi--Jamiolkowski isomorphism, we have to prove that the set of $ABA'B'$ states that commute with conjugations by $U_A\otimes V_B \otimes U^*_{A'}\otimes V^*_{B'}$ coincides with the set of linear combinations of $\mathds{1}_{ABA'B'}$, $\mathds{1}_{AA'}\!\otimes\! \ket{\varepsilon}\!\!\bra{\varepsilon}_{BB'}$, $\ket{\varepsilon}\!\!\bra{\varepsilon}_{AA'}\!\!\otimes\! \mathds{1}_{BB'}$, and $ \ket{\varepsilon}\!\!\bra{\varepsilon}_{AA'}\!\!\otimes \!\ket{\varepsilon}\!\!\bra{\varepsilon}_{BB'}$, i.e. with the tensor product of the two locally invariant subspaces of $AA'$ and $BB$.
This is nothing but a particular instance of a more general phenomenon: given two representations of compact groups $\zeta_i: G_i\rightarrow GL(V_i)$ (with $i=1,2$), the invariant subspace of the \emph{external} tensor product $\zeta_1\boxtimes \zeta_2$ is the tensor product of the two $\zeta_i$--invariant subspaces; in fact, the latter is trivially contained in the former, while at the same time the dimensions are equal, thanks to character theory:
\bq \braket{\chi_1,\, \chi_{\zeta_1\boxtimes\zeta_2}}\ =\ \int dg_1 dg_2\ \Tr[(\zeta_1\otimes\zeta_2)\,(g_1,g_2)]\ =\ \int dg_1 dg_2\ \Tr \zeta_1(g_1)\, \Tr \zeta_2(g_2)\ =\ \braket{\chi_1,\, \chi_{\zeta_1}}\, \braket{\chi_1,\, \chi_{\zeta_2}} . \eq 
In our case, the two groups are $G_1=SU(d_A)$ and $G_2=SU(d_B)$, the two spaces are the set of hermitian operators on $AA'$ and $BB'$, while the representations are defined through $\zeta_1(U)(X_{AA'})\, =\, U_A\otimes U^*_{A'}\ X_{AA'}\ U^\dag_A\otimes U^T_{A'}$ and analogously for $BB'$. \\
Note that for $d_A=d_B$ the separability properties of the Choi states corresponding to \eqref{Phi} have been already considered in \cite{ChrusKoss}. In that paper, the entanglement--breaking conditions for \eqref{Phi} were found under the above simplifying assumption $d_A=d_B$. However, we will see that the general case $d_A\neq d_B$ is way more interesting, because new phenomena such as the existence of PPT entangled states appear. \\
As usual, we will denote by $n$ the minimum between the two local dimensions $d_A, d_B$, i.e. the maximum Schmidt rank of a global pure state. A trick that turns out to be useful in analysing the above $\Phi[\alpha,\beta,\gamma]$ requires the construction of the following family of maps acting on an $n$--dimensional system:
\be
\chi[a,c] \equiv\ \mathds{1}_{n} \Tr +\, a\, D\, +\, c\, I\ ,  \label{chi}
\ee
where $D$ is the projection onto the diagonal, i.e. $D(X)=\mathds{1}\circ X$ and the symbol $\circ$ denotes the Hadamard product.

\begin{example}{The Hadamard Product}{entanglement_2qubit}
The element-wise product between two bipartite density matrices $\rho_{AB}$ and $\sigma_{AB}$ can be stochastically implemented via local operations and classical communication in the following way:
\begin{equation}
\rho_{AB}\circ\sigma_{A'B'}\,=\, \Pi^0_{AA'}\otimes\Pi^0_{BB'}\, \rho_{AB}\otimes\sigma_{A'B'}\, \Pi^0_{AA'}\otimes\Pi^0_{BB'}\, , \label{Had LOCC}
\end{equation}
where $\Pi^0_{AA'}=\sum_{i=1}^{d_A}\ket{ii}\!\!\bra{ii}$ and analogously for $\Pi^0_{BB'}$. A trivial consequence of this observation is the fact that the Hadamard product between separable states written in a local basis is again separable. In appendix \ref{app Schur}, we point out several other properties of the Hadamard product, which we believe are of interest for future research, although they are not directly used in this paper.
\end{example}

\section{Positivity} \label{sec P}

The most direct determination of the positivity region of $\Phi [\alpha,\beta,\gamma ]$ is contained in the proof of Theorem \ref{positivity}. However, we will initially follow a longer path to the main result, passing through Proposition \ref{red pos prob} and Theorem \ref{pos chi}. This is instructive because it shows the link between $\Phi$ and $\chi$ and leads to nice geometrical intuitions. The first important observation we present is indeed the reduction of the problem of positivity for $\Phi[\alpha,\beta,\gamma]$ to the analogous problem for $\chi[a,c]$.

\begin{prop} \label{red pos prob}
The map $\Phi[\alpha,\beta,\gamma]$ is positive iff $\alpha,\beta\geq -1$ and $\chi[\alpha+\beta, \gamma]$ is positive.
\end{prop}

\begin{proof}
Clearly, in order to ensure the positivity of $\Phi$ (parameters are understood for brevity) we have to require that $\Phi (\ket{\Psi}\!\!\bra{\Psi} )\geq 0$ for all global pure states $\ket{\Psi}$. Since $\Phi$ commutes with the conjugation by local unitaries, we can suppose that the input state is in Schmidt normal form, i.e.
\bq \ket{\Psi}\ =\ \sum_{i=1}^n \sqrt{\lambda_i}\, \ket{ii}\, . \eq
Then,
\begin{align}
\Phi (\ket{\Psi}\!\!\bra{\Psi} )\ &=\ \mathds{1}\ +\ \sum_{i,j} (\alpha \lambda_j + \beta \lambda_i) \ket{ij}\!\!\bra{ij}\ +\ \gamma \ket{\Psi}\!\!\bra{\Psi}\ =\\
&=\ \sum_{i\neq j} (1+\alpha \lambda_j + \beta \lambda_i) \ket{ij}\!\!\bra{ij} \ +\ \sum_{ij} \left( \left(1+(\alpha+\beta)\lambda_i \right)\delta_{ij}\ +\ \gamma\, \sqrt{\lambda_i \lambda_j} \right) \ket{ii}\!\!\bra{jj}\ .
\end{align}
From the above block decomposition it is apparent that $\Phi (\ket{\Psi}\!\!\bra{\Psi} )\geq 0$ iff $1+\alpha \lambda_j + \beta \lambda_i\geq 0$ and
\bq \sum_{ij} \left( \left(1+(\alpha+\beta)\lambda_i \right)\delta_{ij}\ +\ \gamma\, \sqrt{\lambda_i \lambda_j} \right) \ket{i}\!\!\bra{j}\ \geq\ 0\ . \eq
Once we require the validity of these conditions for all the probability distributions $\lambda$, the first one reads $\alpha,\beta\geq -1$. At the same time, the above matrix is exactly
\bq \chi[\alpha+\beta,\gamma]\,(\ket{\psi}\!\!\bra{\psi})\ , \eq
with $\ket{\psi}\equiv\sum_i \sqrt{\lambda_i}\ket{i}$, and therefore we are led to require the positivity of the map $\chi[\alpha+\beta,\gamma]$ as the second and last condition. 
\end{proof}

The above proposition reduces our main problem to the question of determining what is the region of positivity of the map \eqref{chi}. Now, we proceed to show how to solve this latter problem.

\begin{thm} \label{pos chi}
The map $\chi[a,c]$ defined by \eqref{chi} is positive iff
\be
\begin{split}
a \geq 0 \qquad &\text{and} \qquad c+\frac{a}{n}+1\, \geq\, 0\ ,\\
&\text{or}\\
-2 \leq a < 0 \qquad &\text{and}\qquad a+c+1\,\geq\, 0\ .
\end{split} \label{pos chi eq}
\ee
\end{thm}

\begin{proof}
Suppose first that $a\geq 0$. Denoting with $D_\lambda$ the diagonal matrix with diagonal $\lambda\in\mathds{R}^n$, the positivity of $\chi[a,c]$ amounts to require that
\be \mathds{1} \,+\, a\, D_\lambda\, +\, c\, \Ket{\psi}\!\!\Bra{\psi}\, \geq\, 0 \qquad \forall\ \Ket{\psi}\, =\, \sum_{i=1}^n \sqrt{\lambda_i}\,\Ket{i}\ ,\quad \lambda\ \text{probability distribution} \label{pos chi pr eq1} \ee
(phases are irrelevant, as the application of a diagonal unitary immediately reveals). Since $\mathds{1}+a\, D_\lambda>0$, after left-- and right-- multiplying by $(\mathds{1}+aD_\lambda)^{1/2}$ the above inequality reads also
\bq \mathds{1}\, +\, c\, \left(\mathds{1}+a D_\lambda\right)^{-1/2}\Ket{\psi}\!\!\Bra{\psi} \left(\mathds{1}+a D_\lambda\right)^{-1/2}\, \geq\, 0\ , \eq
that in turn becomes
\bq \bra{\psi} \left(\mathds{1}+a D_\lambda\right)^{-1} \ket{\psi}\, +\, c\, \left| \bra{\psi} \left(\mathds{1}+a D_\lambda\right)^{-1} \ket{\psi} \right|^2\ \geq\ 0\ , \eq
eventually leading to
\bq 1\, +\, c\, \bra{\psi} \left(\mathds{1}+a D_\lambda\right)^{-1} \ket{\psi}\ \geq\ 0\ , \eq
i.e.
\bq c\ \geq\ -\, \frac{1}{\bra{\psi} \left(\mathds{1}+a D_\lambda\right)^{-1} \ket{\psi}}\ =\ -\, \frac{1}{\sum_i \frac{\lambda_i}{1+a \lambda_i}}\ . \eq
The strictest condition is for the probability distribution $\lambda$ achieving the maximum of $f(\lambda)\, \equiv\, \sum_i \frac{\lambda_i}{1+a\lambda_i}$. From the strict concavity of $\frac{x}{1+ax}$ on $x\geq 0$ it follows easily that $f(\lambda)$ is strictly concave on the simplex of the allowed vectors $\lambda$. Therefore, any internal stationary point is automatically a global maximum. Applying the method of Lagrange multipliers gives us $\lambda_i\equiv \text{const} = \frac{1}{n}$ as the unique stationary point of $f$, and thus 
\bq f_{\max}\, =\, \frac{1}{1+\frac{a}{n}}\ . \eq
The final condition on $c$ becomes exactly
\bq c\ \geq\ -1-\frac{a}{n}\ ,\eq
as stated in \eqref{pos chi eq} for the present case $a\geq 0$. Now, let us consider the opposite possibility, $a<0$. Since the diagonal entries of the operator in \eqref{pos chi pr eq1} are $1+(a+c)\lambda_i$, imposing their positivity for all the allowed $\lambda$ amounts to require that $a+c+1\geq 0$. Moreover, observe that if $a<-2$ it is possible to force $\mathds{1}+a D_\lambda$ to have a $2$--dimensional negative eigenspace. In this case, an addition of pure state could not make the whole operator positive. Therefore, we have to demand $a\geq -2$. In order to show that these two conditions are sufficient, consider an arbitrary pure state $\ket{v}$ and take the matrix element of \eqref{pos chi pr eq1} on it.
\be \bra{v} \left( \mathds{1} + a D_\lambda + c \ket{\psi}\!\!\bra{\psi} \right) \ket{v}\ =\ 1\, +\, a\, \sum_i \lambda_i |v_i|^2\, +\, c\ \Big| \sum_i \sqrt{\lambda_i} v_i\Big|^2 \label{pos chi pr eq2} \ee
In order to prove that the above quantity is always positive if $a\geq -2$ and $a+c+1\geq 0$, we want to formalize the following intuition. On the one hand, it is not possible to make the coefficient of $c$ small without choosing at least two non--zero (and comparable) $v_i$, thus reducing the negative impact of the coefficient of $a$. On the other hand, if one wants to concentrate the weights $\sqrt{\lambda_i} v_i$ on a single element, then $a$ appears always summed with $c$. Let us make the above reasoning rigorous by considering the following geometric inequality, valid for arbitrary $z_1,\ldots,z_n\in\mathds{C}$:
\be \sum_i |z_i|^2\ \leq\ \frac{1}{2}\,\bigg(\Big(\sum_i |z_i|\Big)^2\, +\, \Big|\sum_i z_i\,\Big|^2\bigg)\ . \label{geom ineq} \ee
Its proof follows easily once one rewrites the difference of its two sides as
\bq \sum_{ij} \left(\, |z_i| |z_j| + \Re (z_i^* z_j) - 2 |z_i|^2\delta_{ij}\, \right)\ \geq\ 0\ , \eq
where the latter inequality holds trivially because the left--hand side is a sum of positive terms. A nice geometrical interpretation arises when one considers a closed polygon, that is, numbers $z_1,\ldots, z_n$ such that $\sum_i z_i = 0$. In this case, the above relation states that the sum of the square sides is at most \emph{half} of the square perimeter.

Applying \ref{geom ineq} with $z_i=\sqrt{\lambda_i} v_i$ gives
\bq \sum_i \lambda_i |v_i|^2\ \leq\ \frac{1+k}{2}\ , \eq
where we noticed that $\big|\sum_i \sqrt{\lambda_i} |v_i|\big|\, \leq\, 1$ and wrote for brevity $k\equiv |\braket{\psi | v}|^2$. Then, taking into account that $a<0$, \eqref{pos chi pr eq2} can be lower bounded by
\bq \bra{v} \left( \mathds{1} + a D_\lambda + c \ket{\psi}\!\!\bra{\psi} \right) \ket{v}\ \geq\ 1\, +\, a\, \frac{1+k}{2}\, +\, c\, k\ =\ (1+a+c) k\, +\, \left(1+ \frac{a}{2}\right) (1-k)\ \geq\ 0 \ . \eq 
\end{proof}

\begin{rem}
It is also possible to give a direct proof of the positivity of $\chi[a,c]$ when $-2\leq a\leq 0,\, 1+a+c\geq 0$, based on Hadamard product arguments. Clearly, up to taking convex combinations with known positive maps, one has to prove only the positivity of the extreme point $\chi[-2,1] = \mathds{1}\Tr - 2 D + I$. Then, the crucial observation consists in writing the above map as a Hadamard square of a completely copositive map, i.e. $\chi[-2,1]=\left(\mathds{1}\Tr -I\right)^{\circ 2}$, where the basis we choose from now on to take Hadamard products is the canonical basis of operators $\{\ket{i}\!\!\bra{j}\}_{i,j}$ (or its Choi--Jamiolkowski equivalent $\{\ket{ij}\}_{i,j}$). The following reasoning then shows that Hadamard multiplying two completely copositive maps yields another completely copositive map. Denoting with $R_\phi\equiv (\phi\otimes I)(\ket{\varepsilon}\!\!\bra{\varepsilon})$ the Choi state of $\phi$, one has in fact
\begin{equation*}
R_{\phi_1}^{T_A},\, R_{\phi_2}^{T_A}\, \geq\, 0\quad\Longrightarrow\quad R_{\phi_1\circ \phi_2}^{T_A}\, =\, \left( R_{\psi_1}\circ R_{\phi_2} \right)^{T_A}\, =\, R_{\phi_1}^{T_A}\circ R_{\phi_2}^{T_A}\, \geq\, 0\, .
\end{equation*}
\end{rem}

Putting together Proposition \ref{red pos prob} and Theorem \ref{pos chi} allows us to conclude that $\Phi [\alpha,\beta,\gamma]$ is a positive map iff $\alpha,\beta\geq -1$ and $a=\alpha+\beta,\, c=\gamma$ satisfy \eqref{pos chi eq}. A graphical representation of the region determined by these conditions is given in Figure \ref{Positivity region}.

\vspace{2ex}
\begin{figure}[ht] 
\centering
\includegraphics[height=10cm, width=10cm, keepaspectratio]{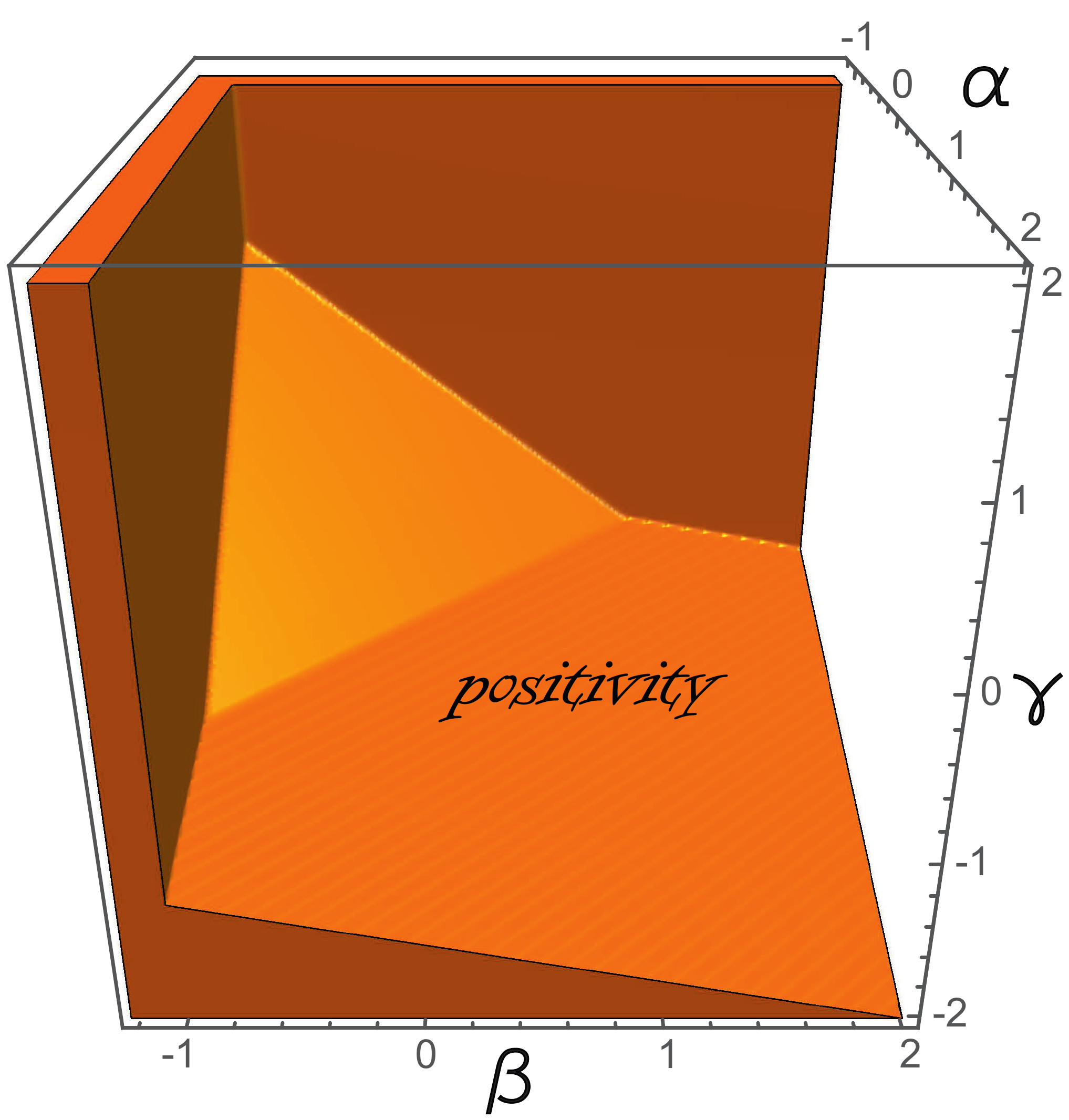}
\caption[]{The solid region represents the parameter region for which the map $\Phi [ \alpha,\beta,\gamma ]$ of \eqref{Phi} is not positive. Here the the $n=4$ case is considered.}
\label{Positivity region}
\end{figure}

However, we want to give a simpler and more direct deduction of the positivity conditions for $\Phi$, as formalized below.

\begin{thm} \label{positivity}
The map $\Phi[\alpha,\beta,\gamma]$ defined by \eqref{Phi} is positive iff:
\be
\alpha,\beta\, \geq -1\qquad \text{and}\qquad \left\{ \begin{array}{lr} \gamma+\frac{\alpha+\beta}{n}+1\ \geq\ 0  & \ \ \text{if}\ \ \alpha+\beta\geq 0\, , \\[2ex] \alpha+\beta+\gamma+1\geq 0 & \ \ \text{if}\ \ \alpha+\beta<0\, . \end{array} \right.
\label{positivity eq} \ee
\end{thm}

\begin{proof}[Alternative proof]
First of all, the above conditions are necessary, as can be seen by choosing suitable input states.
\begin{itemize}
\item A maximally entangled state $\ket{\varepsilon}$ gives $\Phi(\ket{\varepsilon}\!\!\bra{\varepsilon}) = \left(1+\frac{\alpha+\beta}{n}\right) \mathds{1} + \gamma \ket{\varepsilon}\!\!\bra{\varepsilon} \geq 0$, from which $1+\frac{\alpha+\beta}{n}+\gamma \geq 0$ follows.
\item A product state $\ket{11}$ gives $\Phi(\ket{11}\!\!\bra{11}) = \mathds{1} + \alpha\, \mathds{1}\otimes\ket{1}\!\!\bra{1} + \beta\, \ket{1}\!\!\bra{1}\otimes \mathds{1} + \gamma \ket{11}\!\!\bra{11} \geq 0$. Taking the matrix element on the same state produces $1+\alpha+\beta+\gamma\geq 0$.
\item As above, input $\ket{11}$ but take the matrix element on pure states $\ket{21}$ and $\ket{12}$, producing the condition $1+\alpha\geq 0$ and $1+\beta\geq 0$.
\end{itemize}
Moreover, observe that if $\alpha+\beta\geq 0$ the second condition is subsumed under the first one (while the others are trivial), while on the contrary if $\alpha+\beta<0$ the first condition is subsumed under the second one. This completes the proof of the necessity of the inequalities \eqref{positivity eq}.

In order to prove their sufficiency, observe the region that they determine in Figure \ref{Positivity region}. It is apparent that every point belonging to that region is a convex combination of three points on the three following half--lines:
\begin{itemize}
\item the vertical one, at the intersection between the $\alpha\gamma$ and $\beta\gamma$ coordinate planes, $\{\, (-1,-1,\gamma)\,:\ \gamma\geq 1\, \}$ ;
\item the lower one running parallel to the $\beta\gamma$ plane, $\{\, (-1,\beta,-1-\frac{\beta-1}{n})\,:\ \beta\geq 1\, \}$ ;
\item the lower one running parallel to the $\alpha\gamma$ plane, $\{\, (\alpha,-1,-1-\frac{\alpha-1}{n})\,:\ \alpha\geq 1\, \}$ .
\end{itemize}
If we prove that all of the above half--lines are entirely composed of positive maps, we are done. Let us proceed in the same order.
\begin{itemize}
\item $\Phi[-1,-1,\gamma] = (\mathds{1}\Tr-I)\otimes(\mathds{1}\Tr-I) + (\gamma-1) I$ is positive if $\gamma\geq 1$, because the first addend is the tensor product of two completely copositive maps! In other words, up to composing with a global transposition (that is positive and invertible), it is nothing but the product of two completely positive maps.
\item $\Phi\big[ -1,\beta,-1-\frac{\beta-1}{n} \big] = (\mathds{1}\Tr+I)\otimes (\mathds{1}\Tr-I) + (\beta-1)\, I\otimes \big(\mathds{1}\Tr-\frac{1}{n} I\big)$ is positive if $\beta\geq 1$, because:
\begin{itemize}
\item the first addend is positive since it is a tensor product of an entanglement--breaking channel on the first subsystem and a positive map on the second subsystem;
\item the second addend is positive because $\mathds{1}_B\Tr -\frac{I}{n}$ is completely positive if $n=d_B$ and at least $n$--positive if $n=d_A$ (equivalently, one can see directly that $I\otimes \big(\mathds{1}\Tr-\frac{1}{n} I\big)$ is positive by testing it on a pure state).
\end{itemize}
\item The third case is completely analogous to the second one and can be treated in a symmetrically identical way.
\end{itemize}
\end{proof}

\section{Complete Positivity} \label{sec CP}

Determining the range of $\alpha,\beta,\gamma$ for which the map $\Phi[\alpha,\beta,\gamma]$ given by \eqref{Phi} is completely positive requires the construction of the Choi state associated to $\Phi$. Calling $A'B'$ the twin system of $AB$, a maximally entangled state reads
\bq \ket{\mathcal{E}}_{ABA'B'}\ =\ \frac{1}{\sqrt{d_A d_B}}\, \sum_{ij}\, \ket{ij}_{AB}\ket{ij}_{A'B'}\ =\ \ket{\varepsilon}_{AB} \ket{\varepsilon}_{A'B'}\ . \eq
As a consequence, the Choi state becomes
\be R_{\Phi}\ \equiv\ (\Phi_{AB}\otimes I_{A'B'})\, \left( \ket{\mathcal{E}}\!\!\bra{\mathcal{E}}_{ABA'B'} \right)\ =\ \frac{\mathds{1}_{ABA'B'}}{d_A d_B}\, +\, \alpha\, \frac{\mathds{1}_{AA'}}{d_A}\otimes \ket{\varepsilon}\!\!\bra{\varepsilon}_{BB'}\, +\, \beta\, \ket{\varepsilon}\!\!\bra{\varepsilon}_{AA'}\otimes\frac{\mathds{1}_{BB'}}{d_B}\, +\, \gamma\, \ket{\varepsilon}\!\!\bra{\varepsilon}_{AA'}\otimes \ket{\varepsilon}\!\!\bra{\varepsilon}_{BB'} . \label{Choi state Phi} \ee
Since the four addends appearing in the above equation commute, the diagonalization of their sum is straightforward. It is easily seen that the distinct eigenvalues of the operator in \eqref{Choi state Phi} are
\bq \frac{1}{d_A d_B} , \quad \frac{1}{d_A d_B} + \frac{\alpha}{d_A}, \quad \frac{1}{d_A d_B} + \frac{\beta}{d_B}, \quad \frac{1}{d_A d_B} + \frac{\alpha}{d_A} + \frac{\beta}{d_B} + \gamma\ . \eq
Thus, the following theorem is proven.

\begin{thm} \label{complete positivity}
The map $\Phi[\alpha,\beta,\gamma]$ defined by \eqref{Phi} is completely positive iff
\be \alpha\, \geq\, -\, \frac{1}{d_B} \ ,\quad \beta\, \geq\, -\, \frac{1}{d_A} \quad\text{and}\quad 1+d_B\, \alpha + d_A\, \beta + d_A d_B\, \gamma\, \geq\, 0\ .   \label{complete positivity eq} \ee
\end{thm}

As can be easily verified, conditions \eqref{complete positivity eq} imply that every completely positive map in the $\Phi$ class can be written as a convex combination of three points lying on the three half--lines coming out from the vertex $\big(-\frac{1}{d_B},\, -\frac{1}{d_A},\, \frac{1}{d_A d_B} \big)$. They can be represented as follows:
\begin{itemize}
\item the vertical half--line is $\Big\{ \big(-\frac{1}{d_B},\, -\frac{1}{d_A},\, \gamma \big)\, :\ \gamma\geq \frac{1}{d_A d_B} \Big\}$ ;
\item one of the other two is $\Big\{ \big(-\frac{1}{d_B},\, \beta,\, -\frac{\beta}{d_B} \big)\, :\ \beta\geq - \frac{1}{d_A} \Big\}$ ,
\item while the last one is symmetrically described as $\Big\{ \big(\alpha,\, -\frac{1}{d_A},\, -\frac{\alpha}{d_A} \big)\, :\ \alpha\geq - \frac{1}{d_B} \Big\}$ . \end{itemize}

\section{Entanglement--Breaking} \label{sec EB}

The goal of the present section is to answer the question of what is the region in the $\alpha,\beta,\gamma$ parameter space that defines an entanglement--breaking map through equation \eqref{Phi}. Obviously, such a region must be contained in the complete positivity solid defined via \eqref{complete positivity eq}. Reformulating the problem with the help of the Choi--Jamiolkowski isomorphism, we want to determine necessary and sufficient conditions for the separability of the state \eqref{Choi state Phi}. This problem has already been solved in the special case $d_A=d_B$ in \cite{ChrusKoss}, but we will see that the most interesting phenomena appear when one considers the asymmetric case $d_A\neq d_B$.

We will prove the main theorem in two different ways. The first approach (perhaps more elegant, though less direct) employs a provably sufficient set of entanglement witnesses to decide the separability of a state. The second one, instead, uses some positive maps to find necessary conditions for separability, that are then found to be also sufficient by direct construction.

An important observation (that we explain in the box below) exploits the symmetries of our problem to connect this separability question to the problem of testing witnesses \emph{belonging to the same symmetric class}, i.e. the positive maps given by Theorem \ref{positivity}. In what follows, the symbol $\mathbf{CP}$ will stand for the set of completely positive maps (the system they are acting on being understood), while $T$ will denote the transposition map, as usual.

\begin{thm} \label{test EB}
If $d_A=d_B$ then the map $\Phi[\alpha,\beta,\gamma]$ defined by \eqref{Phi} is entanglement--breaking iff it is completely positive and PPT. If $d_A<d_B$, it is entanglement--breaking iff
\be T\Phi\, \in\, \mathbf{CP} \quad\text{and}\quad \bigg( I_A\otimes\Big( \mathds{1} \text{\emph{Tr}}\, -\frac{I}{d_A}\Big)_B \bigg)\ \Phi\, \in\, \mathbf{CP} \label{test EB} \ee
If $d_A>d_B$ a reversed but analogous condition holds (just exchange subscripts $A$ and $B$).
\end{thm}

\begin{proof}
\begin{example}{Separability and Witnesses}{witness}
States which are invariant under a local group action are known to exhibit simplified entanglement properties, whose elementary proofs can be found in Appendix \ref{app sep symm}. Here, we just briefly review some of the most important ones.  Consider two representations $\varphi_1:\mathcal{G}\rightarrow \mathcal{L}\big(\mathcal{H}(n;\mathds{C})\big)$ and $\varphi_2:\mathcal{G}\rightarrow \mathcal{L}\big(\mathcal{H}(m;\mathds{C})\big)$ of a compact group $\mathcal{G}$ on the local spaces of hermitian matrices. Then, the associated projection is
\begin{equation*} \mathcal{P}_\mathcal{G}\ =\ \int_\mathcal{G} dg\ \varphi_1(g)\otimes \varphi_2(g)\, , \end{equation*}
where $\int_\mathcal{G} dg$ is the Haar integral. When acting on a bipartite operator, this projection outputs always a $(\mathcal{G}\otimes \mathcal{G})$--invariant state. Moreover, positivity and separability are preserved under the action of the above superoperator. We remind the reader that in general the Woronowicz criterion states that $\rho$ is separable iff $\Tr\rho W\geq 0$ for all the operators $W$ such that $\Tr \sigma W\geq 0$ for any separable $\sigma$ (such a $W$ is called a separability--witness). When there is a local group action, thanks ultimately to the separability--preserving properties of the group projection, it is possible to give a relaxed version of the above criterion. Namely, a $(\mathcal{G}\otimes \mathcal{G})$--invariant state $\rho$ is separable iff $\Tr \rho W\geq 0$ for all $(\mathcal{G}\otimes \mathcal{G})$--invariant separability witnesses $W$.
\end{example}

As is well known, $\Phi$ is entanglement--breaking iff its Choi state $R_\Phi$ (see \eqref{Choi state Phi}) is separable. The symmetrized version of the Woronowicz criterion we just discussed states that this happens iff $\Tr W R_\Phi\geq 0$ for all the separability--witnesses $W$ that share the same symmetry. In other words, $W$ can be assumed to be the Choi matrix of a positive map $\Phi'\equiv \Phi[\alpha',\beta',\gamma']$ belonging to the set defined by \eqref{Positivity region} (or a limit point of the form $\lim_{M\rightarrow\infty} \frac{1}{M}\,\Phi[M\alpha',M\beta',M\gamma']$, with $(M\alpha',M\beta',M\gamma')$ defining a positive map for all $M>0$, as explained in Section \ref{sec intro}; this case will not be considered further because it does not introduce any new constraint). It is understood that all the operators here act on $ABA'B'$, and that the separability problem is with respect to the $AB|A'B'$ cut.

Instead of requiring the scalar condition $\Tr R_{\Phi'} R_\Phi\geq 0$ for all positive $\Phi'$, one can write
\bq \Tr R_{\Phi'} R_\Phi\, =\, \braket{\mathcal{E} |\, R_{(\Phi')^T\Phi}\, | \mathcal{E}}\, =\, \braket{\mathcal{E} |\, R_{\Phi'\Phi}\, | \mathcal{E}} \eq
and move directly to $\Phi'\Phi\in\mathbf{CP}$ for all positive $\Phi'$. In fact, on the one hand the latter is at least as powerful as the former, while on the other hand the former implies that $\Phi$ is entanglement--breaking, that in turn implies $\Phi'\Phi\in\mathbf{CP}$ for all positive $\Phi'$.

From now on, we assume without loss of generality $d_A\leq d_B$. As detailed in the proof of Theorem \ref{positivity}, all the positive maps in the region defined by \ref{Positivity region} can be written as a convex combination of:
\begin{enumerate}
\item $\Phi[-1,-1,\gamma] = (\mathds{1}\Tr-I)\otimes(\mathds{1}\Tr-I) + (\gamma-1) I$ for some $\gamma\geq 1$;
\item $\Phi\big[ -1,\beta,-1-\frac{\beta-1}{d_A} \big] = (\mathds{1}\Tr+I)\otimes (\mathds{1}\Tr-I) + (\beta-1)\, I\otimes \big(\mathds{1}\Tr-\frac{1}{d_A} I\big)$ for some $\beta\geq 1$;
\item $\Phi\big[ \alpha,-1,-1-\frac{\alpha-1}{d_B} \big] = (\mathds{1}\Tr-I)\otimes (\mathds{1}\Tr+I) + (\alpha-1)\, \big(\mathds{1}\Tr-\frac{1}{d_B} I\big) \otimes I$ for some $\alpha\geq 1$.
\end{enumerate}
Consequently, it suffices to test the above three families of witnesses $\Phi'$. For each family with parameter $x$ ($x=\gamma,\beta,\alpha$ in order), up to taking first the case $x=1$ and then the limit $x\rightarrow\infty$, this is the same as testing the two addends separately. Now, we want to understand which one of the above six tests (two for each one of the three families) can be subsumed under the PPT condition $T\Phi\in\mathbf{CP}$. This happens when the witness is \emph{decomposable}, i.e. it is a positive combination of a CP and a coCP map.
\begin{itemize}
\item[1a.] $(\mathds{1}\Tr-I)\otimes(\mathds{1}\Tr-I)$ is completely copositive because it is the tensor product of two completely copositive maps.
\item[1b.] $I$ is obviously completely positive.
\item[2a.] $(\mathds{1}\Tr+I)\otimes (\mathds{1}\Tr-I)$ is completely copositive because it is the tensor product of two completely copositive maps.
\item[2b.] $I\otimes \big(\mathds{1}\Tr-\frac{1}{d_A} I\big)$ is only positive but not completely positive unless $d_A= d_B$; it will be clear later that it is actually indecomposable whenever $d_A\neq d_B$; for the time being, all we can conclude is that we can not a priori discard this test.
\item[3a.] The same as 2a.
\item[3b.] $\big(\mathds{1}\Tr-\frac{1}{d_B} I\big) \otimes I$ is not the same as 2b, because the condition $d_A\leq d_B$ ensures that $\mathds{1}\Tr-\frac{1}{d_B} I$, and thus the entire map, are completely positive.
\end{itemize}

From the above discussion it should be clear that if $d_A=d_B$ then the PPT test is both necessary and sufficient to ensure that $\Phi$ is entanglement--breaking, while if $d_A<d_B$ the only condition that can not be absorbed in the PPT test is the one in 2b. This is the same as saying that for $d_A<d_B$ a map $\Phi$ of the class in \eqref{Phi} is entanglement--breaking iff:
\begin{itemize}
\item $T\Phi$ is completely positive;
\item and $\Big( I_A\otimes\big( \mathds{1} \Tr\, -\frac{I}{d_A}\big)_B \Big)\ \Phi$ is completely positive.
\end{itemize}

\end{proof}

Now that necessary and sufficient conditions for $\Phi$ to be entanglement--breaking have been written down in the form \eqref{test EB}, it is only a matter of finding out the shape of the corresponding solid. An interesting question, as usual, is whether we really need the second test or on the contrary the PPT condition is actually sufficient. We already saw that the suspected answer is that if $d_A<d_B$ we \emph{do} need the second test. In other words, in that case there are PPT entangled states of the form \eqref{Choi state Phi}.

\begin{thm} \label{EB}
Assume that $d_A\leq d_B$. The map $\Phi[\alpha,\beta,\gamma]$ defined by \eqref{Phi} is entanglement--breaking iff the following conditions are met:
\be
\left\{ \begin{array}{l} \alpha\geq -\frac{1}{d_B}\ ,\quad 1+d_B \alpha+ d_A\beta +d_A d_B \gamma\geq 0\ ,\\[1.5ex]
1-\alpha+\beta-\gamma\geq 0\ ,\quad 1+\alpha-\beta-\gamma\geq 0\ ,\quad 1-\alpha-\beta+\gamma\geq 0\ ,\\[1.5ex]
(d_A d_B -1)(d_A\beta+1)\, -\, (d_B-d_A)(\alpha +d_A\gamma) \, \geq\, 0\ .\end{array} \right.
\label{EB eq} \ee
The last inequality can be omitted if $d_A=d_B$. The solid described by the above system is a double pyramid with triangular basis (see Figure \ref{EB region}). The basis has vertices
\bq \left( -\frac{1}{d_B},\, -\frac{1}{d_B},\, 1\right)\, , \qquad \left( 1,\, -\frac{1}{d_A},\, -\frac{1}{d_A}\right)\, , \qquad \left( -\frac{1}{d_B},\, 1,\, -\frac{1}{d_B}\right)\, , \eq
while the culminating vertices of the two pyramids are
\bq \left( -\frac{1}{d_B},\, -\frac{1}{d_A},\, \frac{1}{d_A d_B}\right)\, , \qquad \left( 1,\, 1,\, 1\right)\, . \eq
\end{thm}

\begin{proof} $ \\ $
Thanks to Theorem \ref{test EB}, we have just to impose the complete positivity of $\Phi$, $T\Phi$ and $\big(I\otimes\left(\mathds{1}\Tr-I/d_A \right) \big)\, \Phi$.
\begin{itemize}
\item $\Phi\in\mathbf{CP}$. This gives the two conditions on the first line together with the requirement $\beta\geq -1/d_A$. However, the latter can be neglected because it follows from the other inequalities of the system \eqref{EB eq}. Indeed, multiplying the second inequality of the first line by $(d_B-d_A)/d_B$ and summing the third line produces exactly $d_A\beta+1\geq 0$.
\item $T\Phi\in\mathbf{CP}$. Taking the partial transposition $T_{A'B'}$ of the Choi state \eqref{Choi state Phi} gives
\bq d_A d_B\, (\Phi_{AB}\otimes T_{A'B'})\, \left( \ket{\mathcal{E}}\!\!\bra{\mathcal{E}}_{ABA'B'} \right)\ =\ \mathds{1}_{ABA'B'}\, +\, \alpha\,\mathds{1}_{AA'}\otimes S_{BB'}\, +\, \beta\, S_{AA'}\otimes \mathds{1}_{BB'}\, +\, \gamma\, S_{AA'}\otimes S_{BB'} ,  \eq
where $S$ indicates the swap operator between two subsystems. Since the four addends in the above equation commute, finding the eigenvalues of their sum is straightforward: they are given by $1+\alpha +\beta +\gamma,\ 1+\alpha -\beta -\gamma,\ 1 -\alpha +\beta -\gamma,\ 1-\alpha -\beta +\gamma$. As is easy to see, $1+\alpha+\beta+\gamma\geq 0$ is already implied by the complete positivity conditions; in fact, using the second inequality of the first line to lower bound $\gamma$ gives
\bq 1+\alpha+\beta+\gamma\, \geq\, 1-\frac{1}{d_A d_B} +\left(1-\frac{1}{d_A}\right)\alpha+\left(1-\frac{1}{d_B}\right)\beta\, \geq\, \left(1-\frac{1}{d_A}\right) \left(1-\frac{1}{d_B}\right)\, \geq\, 0\, , \eq
where we used also $\alpha\geq -1/d_B$ and $\beta\geq -1/d_A$ in the second step. We completed also the second line of \eqref{EB eq}.
\item $\big(I\otimes\left(\mathds{1}\Tr-I/d_A \right) \big)\, \Phi\in\mathbf{CP}$. Imposing this condition requires just mechanical calculations, because the composed map under examination belongs to the same parametric class defined by \eqref{Phi}, and therefore Theorem \ref{complete positivity} applies. Besides the third line of \eqref{EB eq}, we obtain two additional inequalities:
\bq \alpha\, \leq\, \frac{d_A d_B -1}{d_B - d_A}\ ,\qquad \alpha+(d_A d_B -1)\beta + d_A \gamma +d_B-\frac{1}{d_A}\, \geq\, 0\, . \eq
The first one is clearly redundant, because the upper bound $\frac{d_A d_B -1}{d_B - d_A}$ is greater than 1, and the second line of \eqref{EB eq} already ensures $\alpha\leq 1$. As a matter of fact, the second inequality of the above two is also useless, because noting that $\alpha+d_A \gamma\geq -(1+d_A\beta)/d_B$ (thanks to the second inequality of the first line of \eqref{EB eq}) gives us
\bq \alpha+(d_A d_B -1)\beta + d_A \gamma +d_B-\frac{1}{d_A}\, \geq\, \left(d_A d_B - 1 -\frac{d_A}{d_B}\right)\left(\beta+\frac{1}{d_A}\right)\, \geq\, 0\, , \eq
where we used also $\beta\geq -1/d_A$.
\end{itemize}

The shape of the solid defined by \eqref{EB eq} can be seen in Figure \ref{EB region}. Finding the vertices is now an elementary exercise. 
\end{proof}

\begin{proof}[Alternative proof] $ \\ $
It is also possible to give a more direct proof of Theorem \ref{EB}, consisting in finding explicitly separable expressions for the vertices of the solid defined by \eqref{EB eq} and represented in Figure \ref{EB region}. Before entering into the details, we remind the reader that for a bipartite system $AA'$ (with $d_A=d_{A'}$) the isotropic separability--preserving projection acts as
\be \mathcal{P}_{AA'} (\cdot)\ \equiv\ \int dU\ U_A\otimes U^*_{A'}\ (\cdot)\ U_A^\dag\otimes U_{A'}^T\ =\ \ket{\varepsilon}\!\!\bra{\varepsilon} (\cdot) \ket{\varepsilon}\!\!\bra{\varepsilon}\, +\, \frac{\mathds{1}-\ket{\varepsilon}\!\!\bra{\varepsilon}}{d_A^2-1}\ \Tr [(\mathds{1}-\ket{\varepsilon}\!\!\bra{\varepsilon})\,(\cdot)]\, , \label{iso proj} \ee
where $\ket{\varepsilon}$ is a maximally entangled state across $AA'$. Another fact coming from the theory of isotropic states is that
\bq \mathds{1}\Tr +\,I\ ,\qquad \mathds{1}\Tr -\, \frac{I}{d} \eq 
are entanglement--breaking maps when acting on a $d$--dimensional system.

Now we are ready to prove that each of the vertices of the double pyramid of Figure \ref{EB region} corresponds to an entanglement--breaking map (or, equivalently, to a separable Choi operator \eqref{Choi state Phi}). Only one of the five vertices requires a special treatment, while the remaining four are easily seen to correspond to products of entanglement--breaking maps on $A$ and $B$ and therefore to entanglement--breaking maps on $AB$.

\begin{itemize}

\item Basis vertex $\left( -\frac{1}{d_B},\, -\frac{1}{d_B},\, 1\right)$.

The corresponding Choi state reads
\bq R_{\Phi\left[  -\frac{1}{d_B},\, -\frac{1}{d_B},\, 1 \right]}\ =\ \frac{\mathds{1}_{ABA'B'}}{d_A d_B}\, -\, \frac{1}{d_A d_B}\, \mathds{1}_{AA'}\otimes \ket{\varepsilon}\!\!\bra{\varepsilon}_{BB'}\, -\, \frac{1}{d_B^2}\,\ket{\varepsilon}\!\!\bra{\varepsilon}_{AA'}\otimes \mathds{1}_{BB'}\, +\, \ket{\varepsilon}\!\!\bra{\varepsilon}_{AA'}\otimes \ket{\varepsilon}\!\!\bra{\varepsilon}_{BB'}\, . \eq
While being aware that in general $A$ and $B$ do not need to have the same dimension, we can nevertheless introduce a sort of maximally entangled state
\bq \ket{\tilde{\varepsilon}}_{AB}\ \equiv\ \frac{1}{\sqrt{d_A}}\, \sum_{i=1}^{d_A}\, \ket{i}_A\ket{i}_B . \eq
With this notation, some calculations reveal that one can write
\bq R_{\Phi\left[  -\frac{1}{d_B},\, -\frac{1}{d_B},\, 1 \right]}\ =\ \frac{d_A(d_B^2-1)}{d_B}\ \left(\mathcal{P}_{AA'}\otimes \mathcal{P}_{BB'}\right) \left( \ket{\tilde{\varepsilon}}\!\!\bra{\tilde{\varepsilon}}_{AB} \otimes \ket{\tilde{\varepsilon}}\!\!\bra{\tilde{\varepsilon}}_{A'B'} \right) . \eq
Since the right--hand side consist of the application to a $AB|A'B'$--separable state of a map that preserves separability with respect to every cut, we must conclude that the left--hand side is indeed separable.

\item Basis vertex $\left( 1,\, -\frac{1}{d_A},\, -\frac{1}{d_A} \right)$.

The corresponding map reads
\bq \Phi \left[ 1,\, -\frac{1}{d_A},\, -\frac{1}{d_A} \right]\ =\ \mathds{1}\Tr\, +\, \mathds{1}\Tr\otimes I\, -\, \frac{1}{d_A}\, I\otimes \mathds{1}\Tr\, -\, \frac{1}{d_A}\, I\ =\ \left( \mathds{1}\Tr -\frac{I}{d_A} \right)\otimes \left(\mathds{1}\Tr + I\right) ,  \eq
and the rightmost side, being a tensor product of two entanglement--breaking maps on $A$ and $B$, is entanglement--breaking on the composite system $AB$. 

\item Basis vertex $\left( -\frac{1}{d_B},\, 1,\, -\frac{1}{d_B} \right)$.

This case is completely analogous to the previous one:
\bq \Phi \left[ -\frac{1}{d_B},\, 1,\, -\frac{1}{d_B} \right]\ =\ \mathds{1}\Tr\, -\,\frac{1}{d_B} \mathds{1}\Tr\otimes I\, +\, I\otimes \mathds{1}\Tr\, -\, \frac{1}{d_B}\, I\ =\ \left( \mathds{1}\Tr +I \right)\otimes \left(\mathds{1}\Tr - \frac{I}{d_B}\right) .  \eq

\item Culminating vertex $\left( -\frac{1}{d_B},\, -\frac{1}{d_A},\, \frac{1}{d_A d_B} \right)$.

We have
\bq \Phi \left[ -\frac{1}{d_B},\, -\frac{1}{d_A},\, \frac{1}{d_A d_B} \right]\ =\ \mathds{1}\Tr\, -\,\frac{1}{d_B} \mathds{1}\Tr\otimes I\, -\, \frac{1}{d_A} I\otimes \mathds{1}\Tr\, +\, \frac{1}{d_A d_B}\, I\ =\ \left( \mathds{1}\Tr -\frac{I}{d_A} \right)\otimes \left(\mathds{1}\Tr - \frac{I}{d_B}\right) ,  \eq
which is a tensor product of entanglement--breaking maps.

\item Culminating vertex $(1,\, 1,\, 1)$.

The last case is
\bq \Phi[1,\, 1,\, 1]\ =\ \mathds{1}\Tr +\, \mathds{1}\Tr\otimes I\, +\, I\otimes\mathds{1}\Tr +\, I\ =\ \left( \mathds{1}\Tr +I \right)\otimes \left( \mathds{1}\Tr +I \right) , \eq
again a tensor product of entanglement--breaking maps. 

\end{itemize}
\end{proof}

\vspace{2ex}
\begin{figure}[h] 
\centering
\includegraphics[height=8cm, width=8cm, keepaspectratio]{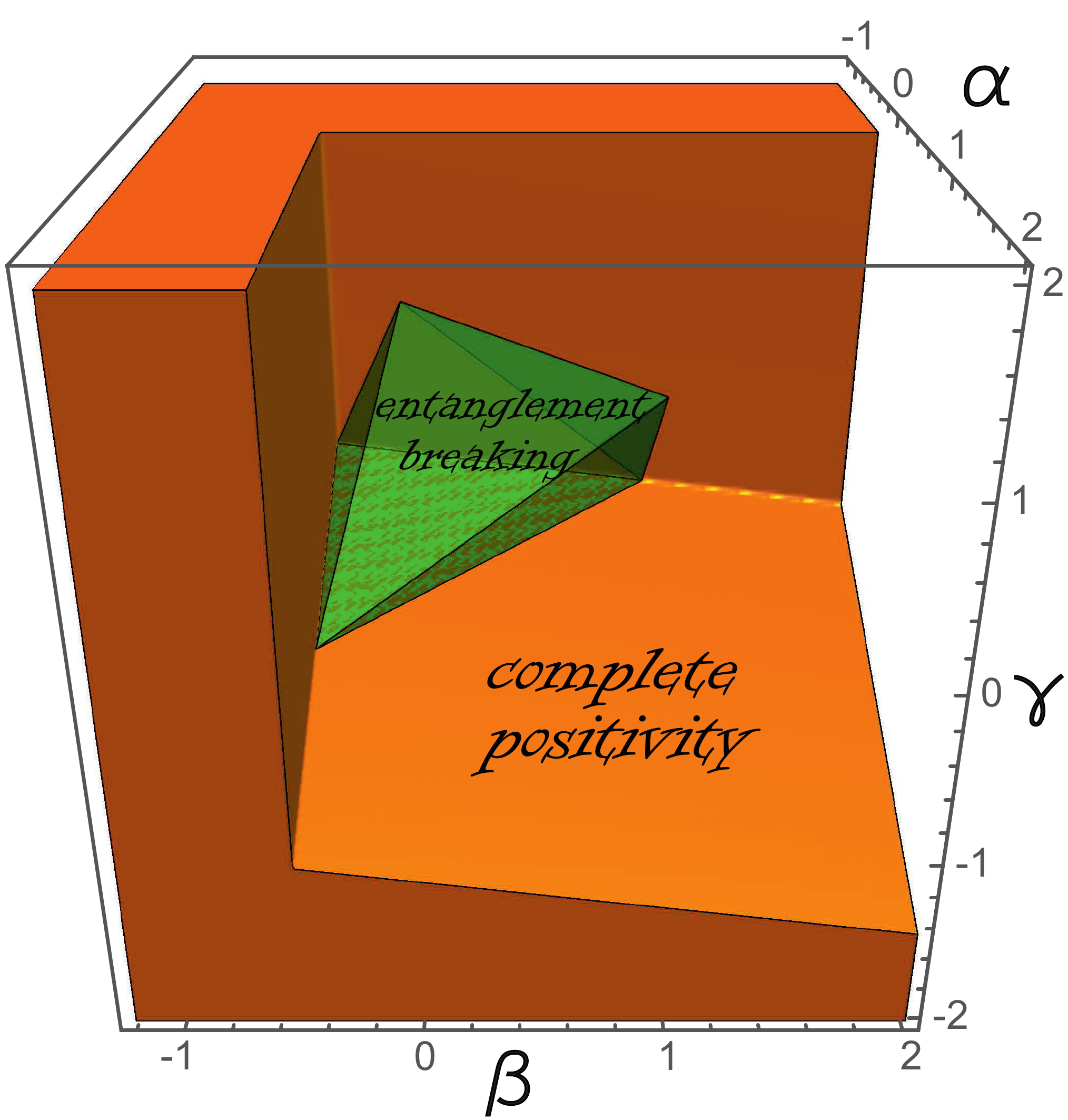}\qquad \includegraphics[height=8cm, width=8cm, keepaspectratio]{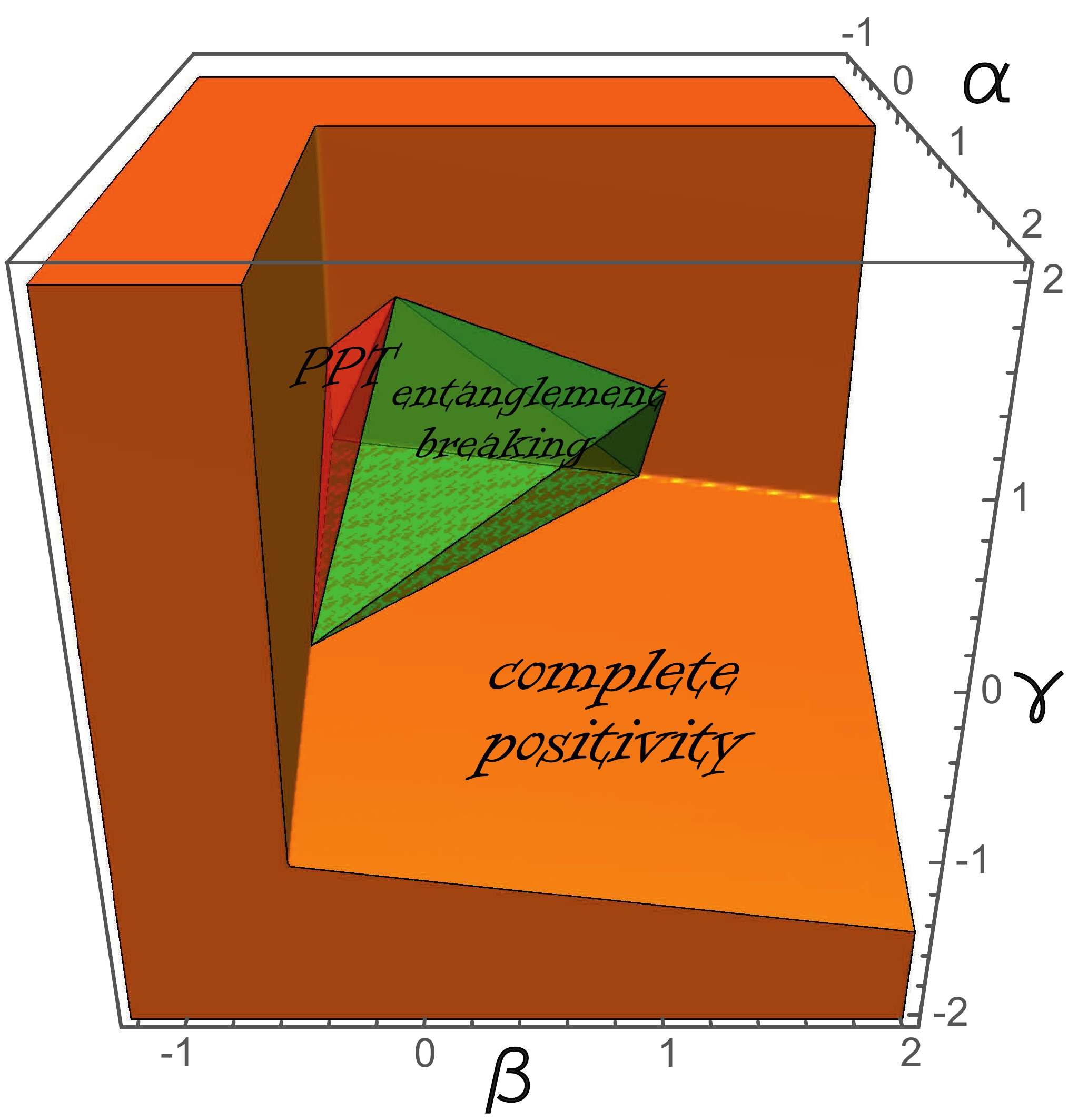}
\caption[]{The three planes identify the complete positivity region for the maps defined in \eqref{Phi} in the $\alpha,\beta,\gamma$ space. Inside that, the entanglement--breaking solid (in green) is shown on the left. On the right, we added the pyramid for which this map is PPT but not entanglement--breaking (in red). Such a region exists iff $d_A\neq d_B$; here we chose the case $d_A=2,\ d_B=6$.}
\label{EB region}
\end{figure}

From the above proofs of Theorem \ref{EB} we learnt that the PPT criterion is not sufficient for deciding separability as soon as $d_A\neq d_B$. Figure \ref{EB region} shows the $\Phi$ maps that are PPT but not entanglement--breaking, forming another pyramid with a face of the entanglement--breaking solid as a basis and the point $\big( -\frac{1}{d_B},\, -\frac{1}{d_A},\, 1-\frac{d_B-d_A}{d_A d_B}\big)$ as the culminating vertex.

\begin{rem}
Several facts of independent interest can be deduced from the above discussion. Let us recall some of them.
\begin{itemize}
\item The state
\be R_{ABA'B'}\ \equiv\ \mathds{1}_{ABA'B'}\, -\, \mathds{1}_{AA'}\otimes \ket{\varepsilon}\!\!\bra{\varepsilon}_{BB'}\, -\, \ket{\varepsilon}\!\!\bra{\varepsilon}_{AA'}\otimes \mathds{1}_{BB'}\, +\, (d_A d_B-d_B+d_A)\, \ket{\varepsilon}\!\!\bra{\varepsilon}_{AA'}\otimes \ket{\varepsilon}\!\!\bra{\varepsilon}_{BB'}\, , \label{PPT ent state} \ee
corresponding to the vertex of the PPT entangled pyramid depicted in red on the right of Figure \ref{EB region}, is a PPT entangled state of the bipartite system $AB|A'B'$.
\item If $d_A<d_B$, the map $I_A\otimes\left(\mathds{1}\Tr-I/d_A \right)_B$ acting on a bipartite system $AB$ is positive but \emph{indecomposable}, as can be seen by noting that it detects the PPT entangled state $R_{ABA'B'}$ defined by \eqref{PPT ent state}. To the extent of our knowledge, this remarkable fact has not been pointed out before. 
\item The map
\bq \Phi\left[  -\frac{1}{d_B},\, -\frac{1}{d_B},\, 1 \right]\ =\ \mathds{1}_{AB}\Tr -\, \frac{1}{d_B}\, \mathds{1}_A\Tr\otimes I_B\, -\, \frac{1}{d_B}\, I_A\otimes \mathds{1}_B \Tr +\, I_{AB} \eq
is entanglement--breaking when acting on a bipartite system $AB$ such that $d_A\leq d_B$. 
\end{itemize}
\end{rem}

\section{Entanglement--Annihilating} \label{sec EA}

An entanglement--annihilating map (see \cite{EA1,EA2,EA3,EA4}) is a positive map $\Xi_{AB}$ acting on a bipartite system $AB$ such that $\Xi_{AB}(R_{AB})$ is a separable state for every input state $R_{AB}$ (as can be easily seen, it suffices to choose pure input states). Observe that this has nothing to do with the notion of entanglement--breaking. While the latter is a completely positive map that always breaks the entanglement between $AB$ \emph{as a whole} and the external world, the former simply breaks the \emph{internal} entanglement between $A$ and $B$.
As extreme examples, on the one hand take the channel $\ket{\varepsilon}\!\!\bra{\varepsilon}_{AB} \Tr$ acting on a bipartite system $AB$ such that $d_A=d_B$. Obviously, such a channel is entanglement--breaking (because it involves tracing away the whole state) but not entanglement--annihilating (because the output state is entangled). On the other hand, consider the map $\mathds{1}_{AB} \Tr - I_{AB}$, that is entanglement--annihilating on $AB$ without even being completely positive. The former observation follows from the fact that $\mathds{1}-\ket{\Psi}\!\!\bra{\Psi}$ is separable for all pure states $\ket{\Psi}$, as it belongs to the Gurvits and Barnum separable ball around the identity \cite{GurvitsBarnum}.

Throughout this section, we want to study the region in the $\alpha,\beta,\gamma$ space such that the corresponding map defined by \eqref{Phi} is entanglement--annihilating. Obviously, there is a naive necessary criterion that must be satisfied: if $\Phi_{AB}$ has to be entanglement--annihilating, then $(I_A\otimes T_B)\, \Phi_{AB}$ must be positive ($T_B$ denotes partial transposition). Maps for which the latter condition holds are called PPT--inducing in \cite{EA4}.

\begin{thm} \label{EA}
Given a map $\Phi[\alpha,\beta,\gamma]$ defined by \eqref{Phi}, the following are equivalent:
\begin{enumerate}
\item $\Phi$ is is entanglement--annihilating ;
\item $\Phi$ is positive and PPT--inducing;
\item $\Phi$ is positive and in addition $\gamma\leq \alpha+\beta+2$;
\item $\alpha,\beta\geq -1$, \ $\gamma+\frac{\alpha+\beta}{n}+1\geq 0$, \ $\alpha+\beta+\gamma+1\geq 0$ \ and \ $\gamma\leq \alpha+\beta+2$.
\end{enumerate}
\end{thm}

\begin{proof} $ \\ $
Let us prove the various implications one by one.
\begin{itemize}
\item[$1.\Rightarrow 2.$] We already saw that an entanglement--annihilating map is necessarily PPT--inducing (and obviously positive).
\item[$2.\Rightarrow 3.$] Input to $\Phi$ the pure state $\ket{\Psi}=\frac{\ket{11}+\ket{22}}{\sqrt{2}}$; the positivity of the partial transpose of the resulting state requires $\gamma\leq \alpha+\beta+2$. 
\item[$3.\Rightarrow 4.$] Trivially obtained by using Theorem \ref{positivity}. The region identified by these conditions is represented in Figure \ref{EA region}.
\item[$4. \Rightarrow 1.$] As one could expect, this is the only point that requires a bit of care. Observing Figure \ref{EA region}, we note that every point of the positive and PPT--inducing region is a convex combination of four points on the four half--lines forming the edges of the set. If we prove that all these four half--lines are composed entirely of entanglement--annihilating maps, we are done.
\begin{itemize}

\item The two lower half--lines have already been studied in the proof of Theorem \ref{positivity}. The one on the right of Figure \ref{EA region} is composed of maps of the form
\bq \Phi\left[ -1,\,\beta,\, -1-\frac{\beta-1}{n} \right]\ =\ (\mathds{1}\Tr+I)\otimes (\mathds{1}\Tr-I)\, +\, (\beta-1)\, I\otimes \Big(\mathds{1}\Tr-\frac{1}{n} I\Big)\, , \eq
with $\beta\geq 1$. Both of the addends of the above equation are entanglement--annihilating, because they are tensor products of a positive and an entanglement--breaking map on the two subsystems. Consequently, their sum is entanglement--annihilating as well. The same reasoning applies to the other lower half--line, that is symmetrically composed of maps of the form $\Phi\big[ \alpha,-1,-1-\frac{\alpha-1}{n} \big]$ with $\alpha\geq 1$.

\item The two upper half--lines of Figure \ref{EA region} are again symmetrically related and can be treated in the same way. The one on the right, for instance, is composed of maps of the form
\bq \Phi\left[ -1,\,\beta,\, \beta+1 \right]\ =\ \Big(\mathds{1}\Tr-\frac{I}{2}\Big)\otimes (\mathds{1}\Tr-I)\, +\, \Big(\beta+\frac{1}{2}\Big)\, I\otimes \left(\mathds{1}\Tr+ I \right)\, , \eq
with $\beta\geq -\frac{1}{2}$. The second addend is clearly entanglement--annihilating because it is a tensor product of the identity and an entanglement--breaking map. Proving that also the first addend is entanglement--annihilating is not completely trivial. Consider an arbitrary pure state $\ket{\Psi}_{AB}=\sum_{i=1}^n \sqrt{\lambda_i}\ket{i}_A\ket{i}_B$, that can be assumed to be Schmidt decomposed in the computational basis without loss of generality (and satisfying $\lambda_1\geq\ldots\geq \lambda_n>0,\ \sum_i \lambda_i = 1$). Define the reduced state $\rho_{\Psi}\equiv \Tr_B\ket{\Psi}\!\!\bra{\Psi}_{AB}= \Tr_A\ket{\Psi}\!\!\bra{\Psi}_{AB}$ and write
\begin{align*}
2\ \Big(\mathds{1}\Tr-\frac{I}{2}\Big)\otimes (\mathds{1}\Tr-I)\, (\ket{\Psi}\!\!\bra{\Psi})\ &=\ 2\, \mathds{1}\, -\, 2\, \mathds{1}\otimes \rho_{\Psi} \, -\, \rho_\Psi\otimes\mathds{1}\, +\, \ket{\Psi}\!\!\bra{\Psi}\ =\\
&=\ \sum_{i\neq j} \sqrt{\lambda_i\lambda_j}\, \ket{ii}\!\!\bra{jj}\, +\, \sum_{i,j} (2 - 2\lambda_j -\lambda_i + \lambda_i \delta_{ij})\, \ket{ij}\!\!\bra{ij}\, .
\end{align*}
Basically, our strategy to prove that the above state is separable will consist in a comparison with a known separable state. Defining $F\equiv\sum_{i\neq j} \ket{ij}\!\!\bra{ij}$, the operator
\be F\, +\, n\, \ket{\varepsilon}\!\!\bra{\varepsilon}\ =\ \sum_{i\neq j} \ket{ii}\!\!\bra{jj}\, +\, \sum_{i,j} \ket{ij}\!\!\bra{ij}  \label{notable state} \ee
is well--known to be separable. A first strategy could be based on a conjugation by a local diagonal matrix $D_\lambda=\text{diag}(\lambda_1,\ldots,\lambda_n)$. One could write
\begin{align*}
D_\lambda \otimes \mathds{1}\ \left( F\, +\, n\, \ket{\varepsilon}\!\!\bra{\varepsilon} \right)\ D_\lambda \otimes\mathds{1}\ &=\ \sum_{i\neq j} \sqrt{\lambda_i \lambda_j}\, \ket{ii}\!\!\bra{jj}\, +\, \sum_{i,j} \lambda_i \ket{ij}\!\!\bra{ij}\ =\\
&=\ 2\, \mathds{1}\, -\, 2\, \mathds{1}\otimes \rho_{\Psi} \, -\, \rho_\Psi\otimes\mathds{1}\, +\, \ket{\Psi}\!\!\bra{\Psi}\, -\\
&-\, \sum_{i,j}\, (2 - 2\lambda_i - 2\lambda_j + \lambda_i \delta_{ij})\, \ket{ij}\!\!\bra{ij}\, .
\end{align*}
If $2 - 2\lambda_i - 2\lambda_j + \lambda_i \delta_{ij}\geq 0$ for all $i,j$ we would be done, because carrying the last addend on the left--hand side of the equation would yield a separable decomposition of the required state. However, the latter inequality fails to hold if $i=j=1$ and $\lambda_1>2/3$. To include also this case, we must think of something different. 

Construct $\ket{\psi}=\sum_i \sqrt{\lambda_i}\, \ket{i}$ and use Theorem \ref{pos chi} to claim that 
\bq A\ \equiv\ \mathds{1} - 2 D_\lambda + \ket{\psi}\!\!\bra{\psi}\ =\ \sum_{i\neq j} \sqrt{\lambda_i\lambda_j}\, \ket{i}\!\!\bra{j}\, +\, \sum_i\, (1-\lambda_i)\, \ket{i}\!\!\bra{i}\ \geq\ 0\, . \eq
Then, define the completely positive map $\zeta_A$ acting as $\zeta_A(X)\equiv A\circ X$, where $\circ$ denotes Hadamard product. One has
\begin{align*}
(I\otimes \zeta_A)(F+n\ket{\varepsilon}\!\!\bra{\varepsilon})\ &=\ \sum_{i\neq j} \sqrt{\lambda_i\lambda_j}\, \ket{ii}\!\!\bra{jj}\, +\, \sum_{i,j}\, (1-\lambda_j)\, \ket{ij}\!\!\bra{ij}\ =\\
&=\ 2\, \mathds{1}\, -\, 2\, \mathds{1}\otimes \rho_{\Psi} \, -\, \rho_\Psi\otimes\mathds{1}\, +\, \ket{\Psi}\!\!\bra{\Psi}\, -\, \sum_{i,j}\, (1-\lambda_i-\lambda_j+\lambda_i \delta_{ij})\, \ket{ij}\!\!\bra{ij}\, .
\end{align*}
Since if $1-\lambda_i-\lambda_j+\lambda_i \delta_{ij}\geq 0$ for all $i,j$ (thanks to $\sum_i \lambda_i = 1$), we can conclude.
\end{itemize}
\end{itemize}
\end{proof}

\vspace{2ex}
\begin{figure}[h] 
\centering
\includegraphics[height=10cm, width=10cm, keepaspectratio]{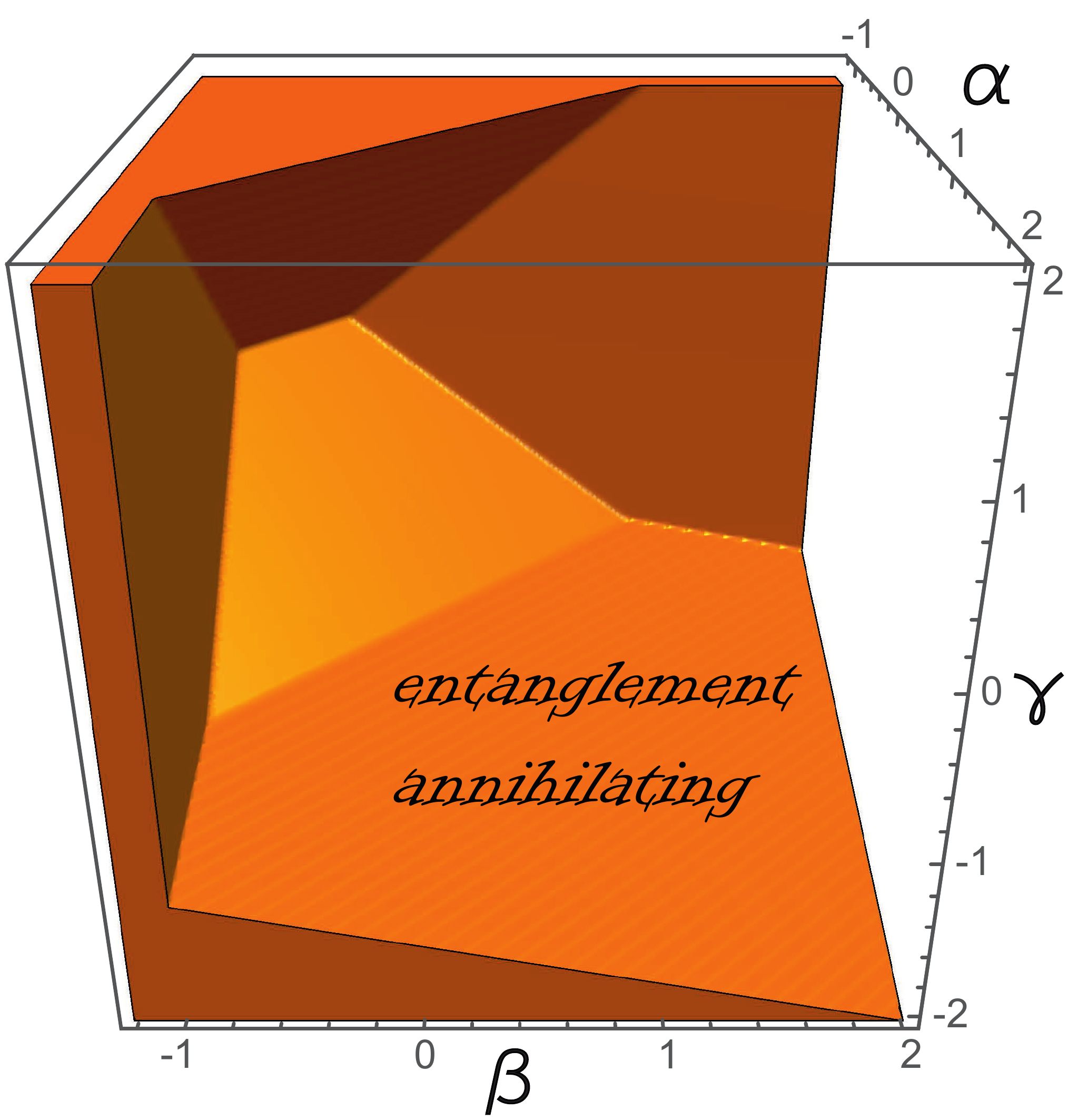}
\caption[]{The convex region outside of the solid represents the parameter range for which the channel \eqref{Phi} is entanglement--annihilating. Here the case $n=4$ is shown. Compare with Figure \ref{Positivity region} and note the extra plane on the top identifying the additional condition $\gamma\leq \alpha+\beta+2$.}
\label{EA region}
\end{figure}

\begin{rem}
We found particularly surprising that the state $ 2\, \mathds{1}\, -\, 2\, \mathds{1}\otimes \rho_{\Psi} \, -\, \rho_\Psi\otimes\mathds{1}\, +\, \ket{\Psi}\!\!\bra{\Psi} $ is separable for all the global pure states $\ket{\Psi}$, especially because of the techniques we employed to prove this fact. We will see in a moment that besides being interesting in itself, is also useful in closing some open problems recently raised in the literature.
\end{rem}

\section{Applications}

Throughout this section, we apply the results in Section \ref{sec EA} to solve some open problems recently posed in \cite{EA3, EA4}. In those papers, the scenario in which a depolarizing channel acts locally on each (equal) side of a bipartite quantum system is considered. More specifically, using the notation $\Delta_q=qI+(1-q)\frac{\mathds{1}}{d} \Tr$ with $-\frac{1}{d^2-1}\leq q\leq 1$, it is asked what is the condition on $q_1,q_2$ under which $\Delta_{q_1}\otimes \Delta_{q_2}$ becomes entanglement--annihilating. In \cite{EA3} (see equation (5)), a \emph{sufficient condition} is found, that reads
\begin{equation} (d^2-1) q_1 q_2\, \leq\, 1\, +\, \frac{(d-2)(d+1)}{d+2}\, (q_1+q_2)\, . \label{FZ eq1} \end{equation}
In particular, for the symmetric case $q_1=q_2=q$, the explicit form
\begin{equation} q\, \leq\, \frac{\,d-2+d\sqrt{\frac{2d}{d+1}}\,}{(d-1)(d+2)} \label{FZ eq1 sym} \end{equation}
is deduced from \eqref{FZ eq1}. At the same time, the condition
\begin{equation}
q\, \leq\, \frac{1+\sqrt{3}}{d+1+\sqrt{3}} \label{FZ eq2}
\end{equation}
is shown to be necessary in order for $\Delta_q\otimes \Delta_q$ to be PPT--inducing. In \cite{EA4} it is also conjectured, that \eqref{FZ eq2} is also \emph{sufficient} for enforcing the PPT--inducing behaviour (if the $q$ range is restricted to the complete positivity interval). Observe that:
\begin{enumerate}[label=\roman*)]
\item there is a gap between the region \eqref{FZ eq1 sym} in which the entanglement annihilation is guaranteed and the region \eqref{FZ eq2} outside which the global map is not even PPT--inducing;
\item in \cite{EA4} it is conjectured that inside that gap the map is still PPT--inducing, while no supposition is made about the entanglement annihilation.
\end{enumerate}

Applying Theorem \ref{EA} straightforwardly solves all these problems, showing that \eqref{FZ eq2} is indeed a \emph{necessary and sufficient} condition for the map $\Delta_q\otimes \Delta_q$ to be PPT--inducing or, equivalently, entanglement--annihilating.  In particular, this proves the conjecture in \cite{EA4}. Before discussing the details of the above solution, we stress that it is not necessary to assume the complete positivity of the local maps (as done in both \cite{EA3} and \cite{EA4}) in order for our problem to make sense. On the contrary, it is enough to demand that the product $\Delta_{q_1}\otimes \Delta_{q_2}$ is positive. Remarkably, this is the case iff the two local maps are themselves positive, because in that case they are either completely positive or completely copositive (or both). The local positivity conditions read
\begin{equation*} -\frac{1}{d-1}\, \leq\, q_1,q_2\, \leq\, 1\, . \end{equation*}
By comparison, remind that the complete positivity conditions are
\begin{equation*} -\frac{1}{d^2-1}\, \leq\, q_1,q_2\, \leq\, 1\, . \end{equation*}
Now, we are in position to show our solution to the aforementioned open questions.

\begin{cor}
The product $\Delta_{q_1}\otimes \Delta_{q_2}$ of local depolarizing maps is entanglement--annihilating if and only if, besides the positivity conditions, the inequality
\begin{equation} (d^2+2d-2) q_1 q_2 \ \leq\ 2\, +\, (d-2) (q_1 + q_2) \label{improv 1} \end{equation}
holds. In the symmetric case $q_1=q_2=q$, \eqref{improv 1} becomes simply
\begin{equation} -\, \frac{\sqrt{3}-1}{d+1-\sqrt{3}} \,\leq\, q\, \leq\, \frac{1+\sqrt{3}}{d+1+\sqrt{3}}\, . \label{improv 1 sym} \end{equation}
\end{cor}

\begin{proof}
We can apply Theorem \ref{EA} in the form of condition 3. As already observed, the product $\Delta_{q_1}\otimes \Delta_{q_2}$ is always positive when the two local maps are positive. Therefore, writing
\begin{equation*}
\Delta_{q_1}\otimes \Delta_{q_2}\ =\ q_1 q_2\ \Phi \left[\frac{d q_2}{1-q_2},\, \frac{d q_1}{1-q_1},\, \frac{d^2 q_1 q_2}{(1-q_1)(1-q_2)} \right]\, ,
\end{equation*}
we have only to impose
\begin{equation*}
\frac{d^2 q_1 q_2}{(1-q_1)(1-q_2)} \ \leq\ \frac{d q_2}{1-q_2}\, +\, \frac{d q_1}{1-q_1}\, +\, 2\, ,
\end{equation*}
which becomes \eqref{improv 1} after elementary algebraic manipulations. Deducing \eqref{improv 1 sym} is now a simple exercise. Observe that the lower bound on $q$ expressed by \eqref{improv 1 sym} is superfluous if the complete positivity condition is imposed (and in that case we end up with \eqref{FZ eq2}), but must be retained if only the positivity is imposed.
\end{proof}

We want to answer another question that is left open in \cite{EA3}. Besides the local depolarizing noise, in that paper a global depolarizing noise of the form $\Delta^{AB}_q=qI_{AB}+(1-q)\frac{\mathds{1}_{AB}}{d_A d_B}\Tr$ is also considered, in the simplest case $d_A=d_B=d$. It is observed that $\Delta_q^{AB}$ is not PPT--inducing if $q>\frac{2}{d^2+2}$, but an explicit entanglement--annihilating construction is provided only for $q\leq\frac{d+2}{(d+1)(d^2-d+2)}$.

It is very simple to observe that $\Delta^{AB}_{2/(d^2+2)}=\frac{1}{d^2+2}(\mathds{1}\Tr +2I)$ is already entanglement--annihilating, indeed. This follows easily from Theorem \ref{EA}, but is also a consequence of the well--known fact \cite{VidalTarrach} that $\mathds{1}+2\rho$ is separable for all normalized density matrices $\rho$ on a bipartite system.

\section{Conclusion}

This concludes the characterization of the extended depolarizing channel. As we have shown, this natural generalisation to bipartite systems of the paradigmatic noisy channel displays a rich structure that nonetheless admits analytical solutions. Besides explicitly working out the parameter regions for which the channel is positive, completely positive, entanglement--breaking and entanglement--annihilating, we have furthermore provided an example of a positive indecomposable map, constructed a notable entanglement--annihilating map and solved some outstanding problems in entanglement annihilation. 
Along the proofs we make some detailed observations on the Hadamard product and the features that the separability problem exhibits in the presence of a symmetry.\\

The maps we consider naturally emerge in situations where two separate systems are subject to white noise and thus provide a useful tool in predicting the physical impact of noise on entanglement. Furthermore they provide new ways to reveal bound entanglement and an analytical characterization of entanglement for a natural class of states. We hope that the observations we point out elucidate certain facts in a way that is useful also for future applications. The solution of the depolarizing channel now provides a promising basis for studying natural generalisations. First of all one should look at straightforward extensions to the multipartite case, which will involve many more parameters, but keep the local unitary symmetries (in fact the number of parameters scales exponentially in the number of parties only). Another potential path to pursue is to study the case of coloured noise. While this would in general lead to a number of parameters growing in the dimension, thermal noise in equidistant Hamiltonians would still give a physically important five parameter family that may yield tractable solutions. \\

We thank Andreas Winter for useful comments and discussion. LL and MH were supported by the Spanish MINECO Project No. FIS2013-40627-P and by the Generalitat de Catalunya CIRIT Project No. 2014 SGR 966. Furthermore, MH acknowledges funding from the Swiss National Science Foundation (AMBIZIONE PZ00P2$\_$161351), with the support of FEDER funds.

\appendix

\section{Entanglement--related properties of the Hadamard product} \label{app Schur}

In general, the Hadamard product (also called \emph{Schur product}) between two matrices $M$ and $N$ is just the element--wise product, that is, $(M\circ N)_{ij}\equiv M_{ij} N_{ij}$. For an introduction to its basic properties, we refer the reader to Chapter 5 of \cite{H&J2}. The main feature of this operation is that the cone of positive matrices is closed with respect to it; in other words, if $M\geq 0$ and $N\geq 0$ then $M\circ N\geq 0$ (this is usually called \emph{Schur product theorem}). Actually, it is easily verified that for a given $A\geq 0$ the \emph{Hadamard channel} $\zeta_A(\cdot)\equiv A\circ(\cdot)$ is not only positive but even completely positive.

Thanks to the Schur product theorem, given two quantum states $\rho,\sigma$, it makes sense to consider the (unnormalized) state $\rho\circ\sigma$. We stress that the definition of Hadamard product is explicitly dependent on the basis we choose to represent the operators $\rho,\sigma$ as matrices. As it turns out, if $\rho_{AB},\sigma_{AB}$ are bipartite states and we fix a local basis to represent them (as we will always do from now on), then the Hadamard product behaves well with respect to the entanglement properties of the input states. Ultimately, this comes from the fact that $\rho_{AB}\circ\sigma_{AB}$ can be obtained stochastically from $\rho_{AB}\otimes\sigma_{A'B'}$ through a local measurement on the bipartite system $AA'|BB$, as \eqref{Had LOCC} shows.

As a first observation, we noticed how this implies that \emph{Hadamard multiplying two separable states yields another separable state}. However, there are examples of weakly entangled states that can be (stochastically) efficiently distilled by taking Hadamard powers. This holds also in the multipartite setting. For instance, the computation described in \cite{GMESchur} (together with some later results from \cite{Buchy}) shows the existence of $(n-2)$--separable $n$--qubit states having a genuinely multipartite entangled Hadamard square.
As a generalisation of the separability--preserving property of the Hadamard product, let us analyse the behaviour of the Schmidt rank function $SR(\cdot)$. Clearly, \eqref{Had LOCC} implies that $SR(\rho\circ \sigma)\leq SR_{AA'|BB'}(\rho_{AB}\otimes\sigma_{A'B'})=SR(\rho) SR(\sigma)$. Now, we proceed to show that this bound can always be saturated, as long as it does not conflict with the trivial requirements $SR(\rho_{AB})\leq\min\{d_A,d_B\}\equiv n$.

Before coming to the details, we remind the reader of some simple facts.
\begin{enumerate}[label=\roman*)]
\item The rank of the $n\times m$ Vandermonde matrix $M_{ij}=\lambda_j^i$ is always equal to $\min\{n,k\}$, where $k$ is the number of distinct elements among the complex numbers $\lambda_1,\ldots,\lambda_m$.
\item The Schmidt rank of a vector $\ket{Z}=\sum_i \ket{v_i}\ket{w_i}$ coincides with the rank of $v^T w$, where $v$ is a matrix defined via $\ket{v_i}=\sum_j v_{ij} \ket{j}$ for some fixed computational basis, and $w$ is constructed analogously. As a first consequence, we deduce that $SR(Z)=r$ whenever $\ket{Z}=\sum_{i=1}^r \ket{v_i}\ket{w_i}$ with $\{\ket{v_i}\}_{i=1}^r,\,\{\ket{w_i}\}_{i=1}^r$ linearly independent families in their own local spaces. As a second corollary, observe that $SR\left(\sum_i \ket{v_i}\ket{v_i^*}\right)=\rk(v)$
\end{enumerate}

Now, consider two positive integers $r,s\leq n$. We have to find two (pure) states $\rho,\sigma$ of Schmidt ranks $r,s$ such that $SR(\rho\circ\sigma)=\min\{n,rs\}$. Choosing a large $N\geq\max\{n,rs\}$, we define $\omega\equiv e^{2\pi i/N}$. Then, for all $0\leq i\leq r-1$ and $0\leq j\leq s-1$, let us introduce the local (unnormalized) states
\begin{equation*} \ket{\alpha_i}\, \equiv\, \sum_{l=0}^{n-1} \omega^{il} \ket{l}\ ,\qquad \ket{\beta_j}\,\equiv\,\sum_{l=0}^{n-1} \omega^{jrl} \ket{l}\, . \end{equation*} 
Using these states as building blocks, we construct
\begin{equation*} \ket{\psi}\, \equiv\, \sum_{i=0}^{r-1} \ket{\alpha_i \alpha^*_i}\ ,\qquad \ket{\varphi}\, \equiv\, \sum_{j=0}^{s-1} \ket{\beta_j \beta^*_j}\, . \end{equation*}
It's easy to see that $\{ \ket{\alpha_i} \}_{0\leq i\leq r-1}$ and $\{ \ket{\beta_j} \}_{0\leq j\leq s-1}$ are two linearly independent families, since all the $\{\omega^i \}_{i=0}^{r-1}$ are distinct, as well as all the $\{\omega^{jr} \}_{j=0}^{s-1}$. Consequently, we see that $SR(\ket{\psi}\!\!\bra{\psi})=r$ and $SR(\ket{\varphi}\!\!\bra{\varphi})=s$. Taking the Hadamard product we obtain
\begin{equation*} \ket{\psi}\!\!\bra{\psi} \circ \ket{\varphi}\!\!\bra{\varphi}\ =\ \ket{\psi\circ\varphi}\!\!\bra{\psi\circ\varphi}\, , \end{equation*}
with
\begin{equation*} \ket{\psi\circ\varphi}\ =\ \sum_{i=0}^{r-1}\sum_{j=0}^{s-1} \ket{\alpha_i\circ\beta_j}\ket{\alpha_i^*\circ\beta_j^*}\, . \end{equation*}
Observe that $\ket{\alpha_i\circ\beta_j}=\sum_{l} \omega^{(i+rj)l} \ket{l}$; therefore, the Schmidt rank of the above state is equal to the rank of the $n\times rs$ matrix $M$ whose entries are given by $M_{l, i+rj}\equiv \omega^{(i+rj)l}$. Since this is a Vandermonde matrix with $rs$ distinct generating numbers, we conclude that $\rk M=\min\{rs,n\}$.

\section{Separability problem in the presence of a symmetry} \label{app sep symm}

Throughout this appendix, we point out some facts that have become part of the common knowledge in the entanglement theory community, and nonetheless we were not able to find clearly stated anywhere. We will be concerned with the behaviour of the separability problem for special classes of states exhibiting certain symmetries. Some of the ideas we review here can be found also in the discussion at the beginning of Chapter 2 of \cite{EggelingPhD}.

First of all, let us fix some notation. The convex cone of separable matrices on a fixed bipartite system will be denoted by $\mathcal{S}$. The central problem is to answer the separability problem for states belonging to a subspace $\mathcal{V}$ of the global space of bipartite hermitian operators $\mathcal{H}(nm;\mathds{C})=\mathcal{H}(n;\mathds{C})\otimes\mathcal{H}(m;\mathds{C})$. We shall always have in mind the paradigmatic example provided by Werner \cite{WernerOrig}, in which $\mathcal{V}$ is composed of all the $(U\otimes U)$--symmetric operators. A definition we find clarifying is the following.

\begin{Def}
Let $\mathcal{V}\subset\mathcal{H}(nm;\mathds{C})$ be a subspace. Then $W\in\mathcal{V}$ is called a \emph{$\mathcal{V}$--separability witness} if
\begin{equation*}
\forall\ \rho\in\mathcal{S}\cap\mathcal{W}\ ,\qquad \text{\emph{Tr}}\, \rho W \geq 0\, .
\end{equation*}
\end{Def}

A well--known fact in standard convex analysis is that any closed convex set can be obtained as the intersection of the closed half--spaces containing it. The first approach to characterize the convex set $\mathcal{S}\cap\mathcal{V}$ could be to apply this theorem directly after the $\mathcal{V}$--section.

\begin{thm} \label{proj HhBh} 
Let $\rho\in\mathcal{V}$ be a bipartite state belonging to the subspace $\mathcal{V}$. Then,
\begin{equation*}
\rho\ \text{is separable}\quad \Longleftrightarrow\quad \text{\emph{Tr}}\,\rho W\geq 0\quad\forall\ \text{$\mathcal{V}$--separability witness}\ W\, .
\end{equation*}
\end{thm}

Observe that a \emph{global} separability witness must be also a $\mathcal{V}$--separability witness. An important fact to realize is that the converse is not true, i.e. a $\mathcal{V}$--separability witness does not need to be a global separability witness. Exactly because of that, restricting the Hahn--Banach theorem does not seem to be a very useful idea. In fact, by doing so we need to find the set of $\mathcal{V}$--separability witnesses, which is in general very different from the set of global witnesses belonging to $\mathcal{V}$. However, a little thought shows that these two sets happen to coincide whenever $\mathcal{V}\cap\mathcal{S}$ is at the same time a section and an orthogonal projection of $\mathcal{S}$ (orthogonality at the level of operators is always intended with respect to the Hilbert--Schmidt product).

\begin{Def}
A subspace $\mathcal{V}\subseteq \mathcal{H}(nm;\mathds{C})$ of a bipartite system is called a \emph{central section} (of $\mathcal{S}$) if there exists a separability--preserving, orthogonal projection of $\mathcal{H}(nm;\mathds{C})$ onto $\mathcal{V}$.
\end{Def}

The above name captures the intuitive idea of how these subspaces look like. For instance, if $\mathcal{S}$ were a sphere in a three--dimensional space, then the only central sections would have been planes passing through its center.

An interesting class of separability--preserving projections includes superoperators of the form $\mathcal{P}_\mathcal{G}=\int_\mathcal{G} dg\ \varphi_1(g)\otimes \varphi_2(g)$, where $\varphi_1:\mathcal{G}\rightarrow \mathcal{L}\big(\mathcal{H}(n;\mathds{C})\big)$ and $\varphi_2:\mathcal{G}\rightarrow \mathcal{L}\big(\mathcal{H}(m;\mathds{C})\big)$ are representations of a compact group $\mathcal{G}$ on the spaces of hermitian matrices, and $\int_\mathcal{G} dg$ is the Haar integral. In this case, $\mathcal{V}$ is the set of fixed points of the representation $\varphi_1\otimes \varphi_2$, that coincides (up to the Choi isomorphism in an orthogonal basis for a $\mathcal{G}$--invariant product) with the set of $\mathcal{G}$--commuting isomorphisms $\mathcal{H}(n;\mathds{C})_{\varphi_1}\rightarrow \mathcal{H}(m;\mathds{C})_{\varphi_2}$, and therefore can be easily determined with the aid of representation theory. Now, the importance of the concept of central section relies primarily on the following elementary observation.

\begin{prop} \label{cent sect prop}
Let $\mathcal{V}\subseteq \mathcal{H}(nm;\mathds{C})$ be a central section. Then any $\mathcal{V}$--separability witness is also a global separability witness.
\end{prop}

\begin{proof}
Consider a $\mathcal{V}$--separability witness $W$, and denote by $\mathcal{P}$ the separability--preserving projection onto $\mathcal{V}$. Since any $\rho\in\mathcal{S}$ satisfies $\mathcal{P}(\rho)\in\mathcal{S}\cap\mathcal{V}$, we have
\begin{equation*}
\Tr W\rho\, =\, \Tr \mathcal{P}(W)\rho \,=\, \Tr W \mathcal{P}(\rho)\, \geq\, 0\, .
\end{equation*}
\end{proof}

As an immediate consequence, we obtain an improved version of Theorem \ref{proj HhBh}.

\begin{thm} \label{impr HhBh}
Let $\rho\in\mathcal{V}\subseteq\mathcal{H}(nm;\mathds{C})$ be a state belonging to a central section $\mathcal{V}$. Then,
\begin{equation*}
\rho\in\mathcal{S}\quad\Longleftrightarrow\quad \text{\emph{Tr}}\, W\rho\geq 0\quad \forall\ \text{global separability witnesses } W\in\mathcal{V}.
\end{equation*}
\end{thm}

\begin{proof}
The implication $\Rightarrow$ is obvious, while the converse follows putting together Theorem \ref{proj HhBh} and Proposition \ref{cent sect prop}.
\end{proof}


\begin{thebibliography}{99}

\bibitem{nielsen} M. A. Nielsen and I. L. Chuang, \textit{Quantum
Computation and Quantum Information} (Cambridge University Press, New York,
2000).

\bibitem{horodeckiqe}  R. Horodecki, P. Horodecki, M. Horodecki and K. Horodecki, Rev. Mod. Phys. {\bf 81}, 865 (2009).

\bibitem{siewert} Ch. Eltschka and J. Siewert, J. Phys. A: Math. Theor. {\bf 47} 424005 (2014).

\bibitem{gurvits} L. Gurvits, {\it Classical deterministic complexity of Edmonds' problem and quantum entanglement} in Proceedings of the thirty-fifth annual ACM symposium on Theory of computing, 10 (2003).

\bibitem{WernerOrig} R. F. Werner, Phys. Rev. Lett. {\bf 40}, 4277 (1989).

\bibitem{Iso-dep} M. Horodecki and P. Horodecki, Phys. Rev. A {\bf 59}, 4206 (1999).

\bibitem{VollbrechtWerner} K. G. H. Vollbrecht and R. F. Werner, Phys. Rev. A {\bfseries 64}, 062307 (2001).

\bibitem{SPA} J. K. Korbicz, M. L. Almeida, J. Bae, M. Lewenstein, and A. Acin, Physical Review A 78, 062105 (2008).

\bibitem{MaxAbelian} D. Chru\'sci\'nski and A. Kossakowski, Phys. Rev. A {\bf 82}, 064301 (2010).

\bibitem{Marcus1} S. M. Hashemi Rafsanjani, M. Huber, C. J. Broadbent and J. H. Eberly, Phys. Rev. A {\bf 86}, 062303 (2012).

\bibitem{axi} C. Eltschka and J. Siewert, Phys. Rev. Lett. {\bf 111} (10) 100503 (2013).

\bibitem{Buchy} L. E. Buchholz, T. Moroder and O. G\"uhne, Ann. Phys. (Berlin), 1521-3889 (2015).

\bibitem{SepBos} N. Yu, arXiv:1501.02957 (2015).

\bibitem{Peres} A. Peres, Phys. Rev. Lett. {\bf 77}, 1413-1415 (1996).

\bibitem{guhnetoth} O. G\"uhne and G. Toth, Physics Reports {\bf 474}, 1 (2009). 

\bibitem{EA3} S. N. Filippov and M. Ziman, Phys. Rev. A {\bfseries 88}, 032316 (2013).

\bibitem{EA4} S. N. Filippov, J. Russ. Laser {\bfseries 35}, 484 (2014).

\bibitem{ChrusKoss} D. Chru\'sci\'nski and A. Kossakowski, Phys. Rev. A {\bf 73}, 062314 (2006).

\bibitem{EA2} S. N. Filippov, T. Rybar and M. Ziman, Phys. Rev. A {\bfseries 85}, 012303 (2012).

\bibitem{EA1} L. Morav{\v c}{\' i}kov{\' a} and M. Ziman, J. Phys. A: Math. Theor. {\bfseries 43}, 275306 (2010).

\bibitem{GurvitsBarnum} L. Gurvits and H. Barnum, Phys. Rev. A {\bfseries 68}, 042312  (2003).

\bibitem{VidalTarrach} G. Vidal and R. Tarrach, Phys. Rev A {\bfseries 59}, 141 (1999).

\bibitem{H&J2} R. A. Horn and C. R. Johnson, \textit{Topics in Matrix Analysis} (Cambridge University Press, Cambridge, 1991).

\bibitem{GMESchur} M. Huber and M. Plesch,  Phys. Rev. A {\bf 83}, 062321 (2011).

\bibitem{Hor x3} M. Horodecki, P. Horodecki and R. Horodecki, Phys. Lett. A {\bf 223}, 1-8 (1996).

\bibitem{EggelingPhD} T. Eggeling, Ph.D. thesis, Braunschweig (2003). Available online at http://d-nb.info/967787947/34 .

\end{thebibliography}
\end{document}